\pdfoutput=1
\PassOptionsToPackage{usenames,dvipsnames}{xcolor}

\documentclass[acmsmall,screen,natbib=false,authorversion]{./acmart}
\RequirePackage[
datamodel=acmdatamodel,
style=acmnumeric, 
]{biblatex}

\hypersetup{hypertexnames=false}

\usepackage{enumitem}
\usepackage{mathtools}
\usepackage{microtype}
\usepackage{todonotes}
\usepackage[capitalize,noabbrev]{cleveref}
\usepackage{tikz}
\usetikzlibrary{positioning,arrows}
\usepackage{stmaryrd}
\usepackage{fontawesome5}
\usepackage{microtype}

\usepackage{thmtools}
\renewcommand{\cref}{\Cref}

\newcommand{\qedNoProof}{\hfill\qedsymbol}

\newtheorem{theorem}{Theorem}
\newtheorem{lemma}[theorem]{Lemma}
\newtheorem{proposition}[theorem]{Proposition}
\newtheorem{corollary}[theorem]{Corollary}
\theoremstyle{definition}
\newtheorem{definition}[theorem]{Definition}

\theoremstyle{remark}
\newtheorem{remark}[theorem]{Remark}

\newtheorem*{claim}{Claim}

\newcommand{\Nat}{\mathbb{N}}
\newcommand{\Prop}{\mathsf{Prop}}

\newcommand{\defeq}{\vcentcolon\equiv}
\newcommand{\Zero}{\mathbf 0}
\newcommand{\One}{\mathbf 1}
\newcommand{\Two}{\mathbf 2}
\newcommand{\Three}{\mathbf 3}
\DeclareMathOperator{\id}{id}
\DeclareMathOperator{\inl}{\mathsf{inl}}
\DeclareMathOperator{\inr}{\mathsf{inr}}

\newcommand{\Subtype}[3]{\{#1 : #2 \ |\ #3\}}

\newcommand{\Ord}{\mathsf{Ord}}
\newcommand{\posalpha}{\alpha_{>\bot}}
\newcommand{\DL}[2]{[{#1},{#2}]_{<}}
\newcommand{\abstrexp}[2]{{#1}^{#2}}
\let\exp\relax
\DeclareMathOperator{\exp}{\mathsf{exp}}
\newcommand{\listexp}[2]{\exp\left({#1},{#2}\right)}
\newcommand{\cons}{::}
\newcommand{\nill}{[\,]}
\DeclareMathOperator{\List}{\mathsf{List}}
\newcommand{\initseg}{\mathrel{\downarrow}}
\let\sup\relax
\DeclareMathOperator{\sup}{\mathsf{sup}}
\let\lim\relax
\DeclareMathOperator{\lim}{\mathsf{lim}}
\DeclareMathOperator{\map}{\mathsf{map}}
\DeclareMathOperator{\isdecreasing}{\mathsf{is-decreasing}}
\newcommand{\denotes}[2]{\llbracket {#1} \rrbracket_{#2}}
\newcommand{\denotesprime}[2]{\llbracket {#1} \rrbracket_{#2}'}
\newcommand{\contoabs}{\mathsf{con}\mbox{-}\mathsf{to}\mbox{-}\mathsf{abs}}
\newcommand{\normalize}{\mathsf{normalize}}

\DeclarePairedDelimiter{\pa}{(}{)}

\newcommand{\LEM}{\textup{LEM}}
\newcommand{\sbullet}{\scalebox{0.5}{$\bullet$}}

\newcommand{\grayson}[2]{\mathop{\mathsf{Gr}}(#1,#2)}

\newcommand{\incr}{\leq_{\textup{\textsf{cl}}}}
\newcommand{\incrst}{<_{\textup{\textsf{cl}}}}
\newcommand{\Id}{\mathsf{Id}}
\crefname{enumi}{Assumption}{Assumptions}

\newcommand{\formalized}{{\color{NavyBlue!75!White}{\raisebox{-0.5pt}{\scalebox{0.8}{\faCog}}}}}
\newcommand{\flinkurl}[1]{\href{#1}{\formalized}}
\newcommand{\baseurl}{https://arxiv.org/src/2501.14542v6/anc/html/Ordinals.Exponentiation.PaperJournal.html}
\newcommand{\flink}[1]{\flinkurl{\baseurl\##1}}
\newcommand{\flinkprime}[2]{\flinkurl{\baseurl\#fixed-assumptions-#1.#2}}

\newenvironment{fequation}[1]
{
  \newtagform{formalized-eq}{(\flink{#1} }{)}%
  \usetagform{formalized-eq}%
  \begin{equation}%
}
{
  \end{equation}%
  \usetagform{default}\ignorespacesafterend%
}

\newenvironment{fequationprime}[2]%
{
  \newtagform{formalized-eq}{(\flinkprime{#1}{#2} }{)}%
  \usetagform{formalized-eq}%
  \begin{equation}%
}
{
  \end{equation}%
  \usetagform{default}\ignorespacesafterend%
}

\setcopyright{cc}
\acmYear{2026}
\acmDOI{}

\acmJournal{JACM}

\begin{document}

\title{Constructive Ordinal Exponentiation}

\author{Tom de Jong}
\email{tom.dejong@ru.nl}
\orcid{0000-0003-1585-3172}

\author{Nicolai Kraus}
\email{nicolai.kraus@nottingham.ac.uk}
\orcid{0000-0002-8729-4077}
\affiliation{%
  \institution{University of Nottingham}
  \city{Nottingham}
  \country{UK}
}

\author{Fredrik Nordvall Forsberg}
\email{fredrik.nordvall-forsberg@strath.ac.uk}
\orcid{0000-0001-6157-9288}
\affiliation{%
  \institution{University of Strathclyde}
  \city{Glasgow}
  \country{UK}}

\author{Chuangjie Xu}
\email{cj-xu@outlook.com}
\orcid{0000-0001-6838-4221}

\renewcommand{\shortauthors}{de Jong, Kraus, Nordvall Forsberg, and Xu}

\begin{abstract}
	In classical mathematics, ordinal arithmetic is well-studied;
	constructively, addition and multiplication are well-known.
        We study constructive ordinal exponentiation, where ordinals are taken to be transitive, extensional, and wellfounded orders on sets.
	We present a negative result showing that general constructive ordinal exponentiation is impossible, but we suggest two definitions that come close.
	The first definition is abstract and solely motivated by the expected equations; this works as long as the base of the exponential is positive.
	The second definition is based on decreasing lists and can be seen as a constructive version of Sierpi{\'n}ski's definition via functions with finite support; this requires the base to have a trichotomous least element.
	Whenever it makes sense to ask the question, the two constructions are equivalent, allowing us to prove algebraic laws, cancellation properties, and preservation of decidability of the exponential.
	The core ideas do not depend on any specific constructive set theory or type theory, but a concrete computer-checked mechanization using the Agda proof assistant is given in homotopy type theory.
\end{abstract}



\fancypagestyle{firstpagestyle}{
  \fancyfoot{}
  \fancyhead{}
}
\fancyfoot{}

\maketitle

\section{Introduction}
In classical mathematics, ordinals have rich and interesting structure.
How much of this structure can be developed in a constructive setting?
Classical ordinals have powerful applications such as tools for establishing consistency of logical
theories~\cite{Rathjen2007}, proving termination of processes~\cite{Floyd1967}, and justifying induction and recursion~\cite{Aczel1977,DybjerSetzer1999}, which would all be valuable to have in constructive mathematics and, for example, in proof assistants based on constructive type theory.
In this article, we explore the structure of ordinal arithmetic constructively, with an emphasis on the most difficult case of ordinal exponentiation.

\subsection{Ordinal Arithmetic and the Contributions of this Paper}

Ordinal arithmetic generalizes natural number arithmetic, so it is useful to consider the latter for a moment.
If we think of a natural number as a finite set, it is natural to define the sum of two natural numbers as the disjoint union of these sets; however, if we think of Peano natural numbers that are constructed from $0$ and a successor operation $S$, then an obvious recursive definition of $m + S(n)$ is given by $S(m + n)$.
Thus, as explained e.g.\ by Halmos in \emph{Naive Set Theory}~\cite[Sec.~21]{Halmos1960}, there are two standard approaches to ordinal arithmetic.
Halmos refers to the two approaches as \emph{set theoretic} (constructions on
sets) and \emph{recursive} (definitions by transfinite recursion on one
argument).
In the classical literature on ordinal numbers, both approaches are used.
The following recursive definitions of addition and multiplication are standard:
\[\label{eq:addition-multiplication-spec}
\begin{aligned}
	\alpha + 0 &\coloneqq \alpha\\
	\alpha + (\beta+1) &\coloneqq  (\alpha + \beta) + 1 \\
	\alpha + \lambda &\coloneqq \sup_{\beta < \lambda}\left(\alpha + \beta\right) &&\text{(if \(\lambda\) is a limit)} \\[0.5em]
	\alpha \times 0 &\coloneqq 0\\
	\alpha \times (\beta + 1) &\coloneqq  (\alpha \times \beta) + \alpha \\
	\alpha \times \lambda &\coloneqq \sup_{\beta < \lambda}\left(\alpha \times \beta\right) &&\text{(if \(\lambda\) is a limit)}
\end{aligned}
\tag{\(\ast\)}
\]
Similarly, exponentiation can classically be defined by:
\[\label{eq:intro-spec}
\begin{aligned}
	\alpha^0 &\coloneqq 1\\
	\alpha^{\beta+1} &\coloneqq  \alpha^\beta \times \alpha \\
	\alpha^{\lambda} &\coloneqq \sup_{\beta < \lambda}\alpha^{\beta} && \text{(if \(\lambda\) is a limit, \(\alpha \neq 0\))} \\
	0^\beta &\coloneqq 0  && \text{(if \(\beta \neq 0\))}
\end{aligned}
\tag{\(\dagger\)}
\]
A functional programmer, for whom constructivity is naturally non-negotiable, will recognize the \emph{zero} and \emph{successor} cases for the implementation of natural number arithmetic (see, for example, Hutton~\cite{hutton2016programming}).
When moving from natural to ordinal numbers, the situation becomes more delicate, and much depends on the very definition of what ordinal numbers are.
If we took an ordinal to be a certain ``syntax tree,'' we could view the above definitions as the straightforward extension of the relevant equations for the natural numbers by a case for limit ordinals; the above choice ensures that each operation is continuous in the second argument.
However, such ``syntactic'' presentations are merely notation systems, and the traditional definition of an ordinal is that of a well ordered set with certain properties.
Classically, the distinction is not too important because it is possible to switch between different representations (as long as the ordinals involved are sufficiently small to have a notation in the notation system); constructively, the distinction is crucial.
Being able to make a case distinction on whether a well ordered set is zero, a successor, or a limit, is available if and only if the law of excluded middle holds~\cite[Thm~63]{KNFX2023}.
Therefore, the above equations are, constructively, not a valid definition;
however, they are still a reasonable \emph{specification} of how the operations should behave.
(The definition of ordinals as well as alternative approaches are discussed in \cref{sec:ord-in-hott} below.)

For addition and multiplication, there are well known constructions using what Halmos~\cite{Halmos1960} calls the set-theoretic approach. As suggested above, addition uses disjoint unions, and it is not difficult to define multiplication via cartesian products of sets.
These are well behaved constructively --- for example, one can prove that they satisfy the specification \eqref{eq:addition-multiplication-spec} above.
On the other hand, subtraction (and, similarly, division) is inherently non-constructive for the approach of ordinals as well ordered sets~\cite[{\texttt{Ordinals.AdditionProperties}}]{TypeTopologyOrdinals}.
For exponentiation, most classical textbooks on ordinals~\cite{Halmos1960,Enderton1977,Rosenstein1982,Holz1999,Jech2003,Pohlers2009} simply employ the non-constructive definition by cases.
Notable exceptions are the classical monograph by Sierpi{\'n}ski~\cite{sierpinski} and the constructive work by Grayson~\cite{Grayson1978,Grayson1979,Grayson1982}, which we will return to shortly.
In fact, going back all the way to Cantor's original writings on ordinals~\cite{Cantor1897}, addition and multiplication are defined ``explicitly'' in \S 14, whereas exponentiation (with base \(\alpha > 1\)) is defined by case distinction and transfinite recursion in \S 18.

Perhaps one reason why constructive ordinal exponentiation has so far been underdeveloped is that the operation is somewhat non-intuitive, even in a classical setting.
For example, for the first infinite ordinal \(\omega = \sup_{n : \Nat}n\), continuity in the exponent (cf.~the specification \eqref{eq:intro-spec}) implies that \(2^\omega = \omega\) --- an equation very different from what one might expect from, for example, cardinal exponentiation.
This in turn rules out an understanding of ordinal exponentiation in terms of category theoretic exponentials: even though the collection of ordinals forms a posetal category where the morphisms are functions with additional structure, exponentiation is not right adjoint to multiplication. 
Even worse, as we will prove in \cref{thm:no-go-exp}, any exponentiation operation satisfying the specification \eqref{eq:intro-spec} is inherently non-constructive: such an operation can be shown to exist if and only if the law of excluded middle holds.
However, we will show that we can still define exponentiation \(\alpha^\beta\) whenever \(\alpha \geq 1\).

In this paper, we present
two approaches to constructive ordinal exponentiation:
\begin{enumerate}[label=\Roman*]
\item\label{item:abstract-def-of-exp} An abstract construction based on suprema of ordinals and transfinite recursion, which can be proven to satisfy the specification \eqref{eq:intro-spec} when \(\alpha \geq 1\), i.e., when \(\alpha\) has a least element (\cref{abstract-exp-satisfies-spec}).
\item\label{item:concrete-def-of-exp} A concrete construction via decreasing lists, inspired by Sierpi\'nski's~\cite{sierpinski} classical construction. We show that it satisfies the specification \eqref{eq:intro-spec} when \(\alpha\) has a \emph{trichotomous} least element, i.e., an element \(\bot \in \alpha\) such that either \(x = \bot\) or \(\bot < x\) for every element \(x \in \alpha\) (\cref{concrete-exp-satisfies-spec}).
\end{enumerate}

These two approaches have different advantages.
For example, the abstract construction
allows convenient proofs of algebraic laws such as \(\alpha^{\beta + \gamma} = \alpha^\beta \times \alpha^\gamma\) using the universal property of the supremum as a least upper bound, while the concrete definition using decreasing lists is easily seen to preserve properties such as trichotomy and decidable equality.
Our main theorem is that the two approaches are actually equivalent (\cref{abstrexp-listexp-coincide}).
This allows us to freely switch between the representations, not unlike how it is done in classical ordinal theory, and obtain various further results.
In a classical setting, where the specification \eqref{eq:intro-spec} fully defines exponentiation, both our constructions are necessarily equivalent to the usual definition by case distinction, and everything we do works for this standard definition.
Therefore, our development is a generalization of, rather than an alternative to, the existing theory of classical ordinal exponentiation.

\subsection{Constructive Foundations}\label{sec:foundations}

There is a variety of constructive set theories and constructive type theories, and it is open whether and how we can translate between them in general.
Much of mathematics is sufficiently foundation-agnostic to work in many different settings, and the core ideas presented in this paper are no exception; it should in principle be possible to carry them out in different constructive settings, as long as appropriate versions of certain basic constructions are available, cf.~\cref{subsect:setting}%
\footnote{With respect to foundations, we believe that the status of our constructions is similar to Grayson's work on constructive well-orderings, of which the very first sentence reads: \emph{The constructions and proofs of this paper are to be understood as taking place in some kind of basic set theory or type theory, based on intuitionistic logic}~\cite{Grayson1982}.}.
However, as it often happens, the concrete formulation of the results will depend on the chosen foundation. As an example, if we work in a constructive set theory, it would be absurd to expect arithmetic laws such as $(\alpha + \beta) + \gamma = \alpha + (\beta + \gamma)$ to hold as equalities on the nose; as each ordinal has an underlying set, and the best we can possibly hope for is an order-preserving isomorphism between sets.
Depending on the concrete setting, even this may be too ambitious. For instance, some of the constructions used below, such as suprema, involve quotients, which in intensional type theories without quotient types may instead be represented using setoids~\cite[Section~5.2]{Setoids}. Thus, for the constructions to go through, the underlying sets of ordinals may have to be replaced by setoids, with the corresponding notion of order isomorphism formulated at the level of setoids.

Notwithstanding the foundation-agnosticism claimed above, and in order to formulate precise statements, we choose a concrete foundation for the current work, namely \emph{homotopy type theory (HoTT)}~\cite{HoTTBook}.
The benefit of this foundation is that all standard notions work out of the box. Sets can be taken to be the usual sets (or \emph{h-sets}) of homotopy type theory, and all equalities can be interpreted as the internal identity type, as usual written using $=$, the standard equality symbol.
Essentially thanks to univalence, there is no need to replace sets with setoids or equality with setoid isomorphism. While something along these lines will happen automatically if homotopy type theory is interpreted using models in other foundations, this is not a concern for the purposes of the current work.

From this point of view, homotopy type theory offers a convenient framework for our development.
Especially the abstract construction of exponentials (Approach \ref{item:abstract-def-of-exp} above) benefits, as it makes use of set quotients and the univalence axiom; the concrete list-based construction (Approach \ref{item:concrete-def-of-exp}) does not rely on these concepts for its definition, but still requires them to prove its properties.
Then, after proving the equivalence of both formulations,
univalence (representation independence) allows us to transport properties along this equivalence and get the best of both worlds.
This may be seen as an example where univalence helps for work that happens purely on the level of sets.

A further benefit of choosing a Martin-L\"of style type theory as our concrete foundation is that we were able to fully mechanize all results of this paper in the proof assistant Agda~\cite{agda}.
This again would have been possible in many flavors of type theory, but homotopy type theory makes the development very natural and, we hope, readable.
The details of this mechanization are given below in \cref{subsect:agda}.

We believe that the reader will be able to adapt the results in this paper to other constructive foundations of interest such as the Calculus of Inductive Constructions/Rocq~\cite{Rocq} or a version of constructive set theory, as long as the foundation allows the basic constructions of predicative Martin-L\"of style type theory (cf.~\cref{subsect:setting}).
As discussed above, the canonical way to remove the reliance on homotopy type theory principles is to weaken the notion of equality and switch to setoid isomorphisms. This would lead to additional boilerplate, but we do not anticipate any further difficulties.

\subsection{Related Work}

\subsubsection*{Ordinals in homotopy type theory}

In the context of homotopy type theory, the definition of ordinals as transitive, extensional, wellfounded orders was suggested in the Homotopy Type Theory Book~\cite{HoTTBook}.
Escard\'o~\cite{TypeTopologyOrdinals} developed a substantial Agda formalization of ordinal theory with many new results, on which our formalization is based.
We are also building on our own previous work~\cite{KNFX2023},
where we gave implementations of addition and multiplication for the ordinals we consider in this paper, but notably no implementation of exponentiation.

In our previous work~\cite{KNFX2023,krausNordvallforsbergXu:ordinals}, we further compared the notion of ordinals considered here with other notions of constructive ordinals, namely Cantor Normal Forms~\cite{Cantor1897,troelstra-schwichtenberg} and Brouwer trees~\mbox{\cite{church:1938,kleene:notation-systems}}.
In a classical setting, these are simply different representations for ordinals, and one can (as long as it makes sense size-wise) convert between them. Constructively, the different notions split apart as there is a trade-off between decidability and the ability to calculate unrestricted limits or suprema. While Cantor Normal Forms enjoy excellent decidability properties, they only allow the calculation of finite limits or suprema. For Brouwer trees, some properties are decidable, but only very specific infinite limits can be calculated. The ordinals considered in this paper enjoy no decidable properties but allow the formation of arbitrary (small) suprema. They are also the most general ones, in the sense that Cantor Normal Forms and Brouwer trees (as well as other notions of ordinals) can be viewed as a subtype of the type of transitive extensional wellfounded orders~\cite{KNFX2023}.

\subsubsection*{Exponentiation as lists}

In a classical setting, the realization of ordinal exponentiation as finitely supported functions is well known.
Further, in such a setting, the implementation of exponentiation~\(\omega^\beta\) with base~\(\omega\) as decreasing lists of ordered pairs \((n, b)\) is also known to proof theorists working on ordinal analysis.
For example, we first came across this idea in Setzer's survey article~\cite{Setzer2004}.
Hancock~\cite{Hancock2008} discusses the cases of \(2^\beta\) and \(\omega^\beta\), but admits that the definition in terms of finite support is ``rather non-constructive, and admittedly rather hard to motivate''.

A construction of ordinal exponentiation as decreasing lists for an arbitrary base \(\alpha\) was suggested by Grayson in his PhD thesis~\cite{Grayson1978}, with the relevant part published as~\cite{Grayson1982}.
Compared to our construction, Grayson's does not include the condition that the base has a trichotomous least element.
Grayson's construction, which comes without any proofs, is thus significantly more general --- but also, unfortunately, incorrect in this generality. The special case that works is equivalent to our suggestion (cf.~\cref{sec:grayson}).
Grayson further claims that his construction satisfies a ``recursion equation'',
which cannot hold in the generality claimed. However, when read as a recursive
\emph{definition}, it yields precisely our abstract construction with the caveat
that Grayson uses setoids which we avoid thanks to univalence.

\subsubsection*{Mechanization of ordinal exponentiation}

As far as we know, this work is the first that comes with a mechanization of constructive ordinal exponentiation in a proof assistant.
However, others have done so in a classical logic, where the definition is considerably more straightforward.
For example, the Lean mathematical library~\cite{mathlib} defines ordinal exponentiation by the case distinction~\eqref{eq:intro-spec}, while Blanchette, Popescu and Traytel~\cite{Blanchette2014} use classical logic to encode exponentials as functions with finite support in Isabelle/HOL.
However, we emphasize that we do not give a different, or ``yet another'' definition in this paper --- in a classical setting, it is not hard to show that all these definitions coincide, including ours.
Rather, we give a definition that improves on the existing ones, in the sense that it is well behaved in the absence of classical logic and (depending on the setting) may be better from the point of view of computation.

\subsection{Outline of the Paper}

We start by recalling how ordinals are defined in homotopy type theory in \cref{sec:preliminary}, where we also discuss their basic properties and introduce a precise specification of exponentiation.
In \cref{sec:abstract-approach}, we give a first constructive implementation of exponentiation \(\abstrexp{\alpha}{\beta}\) using an abstract approach via suprema, and prove it well behaved when the base is positive, i.e., when \(\alpha \geq 1\) (\cref{abstract-exp-satisfies-spec}).
In \cref{sec:concrete-approach}, we give a second constructive implementation via decreasing lists, which is more concrete.
We show that it is well behaved when the base has a trichotomous least element (\cref{concrete-exp-satisfies-spec}).
In \cref{sec:equivalence}, we then compare the two approaches:
in \cref{abstrexp-listexp-coincide}, we show that they are in fact equivalent when the base has a trichotomous least element (and thus in particular is positive).
Further, we explain how the concrete implementation in the form of decreasing lists can be seen as a type of normal forms for the abstract implementation (\cref{related-denotations}), and discuss Grayson's~\cite{Grayson1982} similar suggestion of ordinal exponentiation via decreasing lists and the connection with our work.
In \cref{sec:cancellation}, we consider cancellation properties such as $\alpha + \beta \leq \alpha + \gamma$ implying $\beta \leq \gamma$. While these are straightforward to prove using classical logic, we give a more abstract argument which is also valid constructively (\cref{thm:cancellation}).
In \cref{sec:taboos}, we explore what is not possible to achieve constructively for ordinal exponentiation, by giving several properties that hold in the classical theory of ordinals, but are in fact all constructively equivalent to the law of excluded middle.
For example, it is not possible constructively to define operations such as subtraction, division and logarithm.
In \cref{sec:approximation}, however, we construct ``approximate'' versions of such operations (\cref{thm:approximation}).

\subsection{History of this Paper}

This paper is an extended journal version of a paper presented at LICS 2025~\cite{exponentiation-lics}. Compared to the conference version, we have emphasised that our development should be relevant more generally for constructive approaches to ordinals, and is not just limited to homotopy type theory. In \cref{sec:preliminary}, we also discuss specifications for addition and multiplication, and the newly added \cref{sec:cancellation,sec:approximation} discuss cancellation properties and approximate inverse functions such as approximate subtraction, division and logarithm for ordinal arithmetic more generally.

\subsection{Formalization}\label{subsect:agda}

All our results have been formalized in the Agda proof assistant, and type check using Agda 2.8.0 and 2.9.0-rc1.
Our formalization~\cite{formalization} is building on, and part of the TypeTopology development~\cite{TypeTopologyOrdinals} by Escard\'o and contributors.
Our Agda code is available on GitHub at
\url{https://github.com/martinescardo/TypeTopology} (commit: 31b30d2) and
archived on SoftwareHeritage as
\href{https://archive.softwareheritage.org/swh:1:dir:52593e9c144b1949278275032fe55c49b8dcfca9}{swh:1:dir:52593e9c144b1949278275032fe55c49b8dcfca9}.
A~browseable HTML rendering of all results in this paper is available at
\url{\baseurl}.
Throughout the paper, the symbol \formalized{} is a clickable link to the
corresponding machine-checked statement in that HTML file.

We found Agda extremely valuable when developing the decreasing list approach to
exponentiation, as its intensional nature makes for rather combinatorial
arguments that we found challenging to rigorously check on paper.

\section{Ordinals in a Constructive Setting}\label{sec:preliminary}

In this section, we introduce ordinals in the setting of homotopy type theory, discuss the well known constructions of addition and multiplication, and what we expect exponentiation to satisfy.

\subsection{Setting and Notation}\label{subsect:setting}
As explained in \cref{sec:foundations}, our work can be interpreted in many constructive foundations, but concretely, we work in intensional Martin-L\"of type theory extended with the univalence axiom and set quotients.
We also use function extensionality tacitly, as it follows from the univalence axiom.
In particular, everything we do works in homotopy type theory as introduced
in the ``HoTT book''~\cite{HoTTBook}.
Some of the constructions work in a more minimalistic setting, and our Agda formalization (but not this paper) tracks assumptions explicitly.
The univalence axiom is not only used to conveniently transport properties between equivalent representations, but is also crucial for proving equations of ordinals.

Regarding notation, we mostly follow the aforementioned book.
In particular, the identity type is denoted by \(a = b\), while definitional (a.k.a.\ judgmental) equality is written \(a \equiv b\).
In the paper, we keep universe levels implicit; they are explicitly tracked in the Agda formalization.
Recall that a type \(A\) is called a \emph{proposition} if it has at most one inhabitant, i.e., if \(a = b\) holds for all \(a, b : A\).
A~type \(A\) is called a \emph{set} if all its identity types $a = b$ for \(a, b : A\) are propositions.
If \(P : A \to \Prop\) is a family of propositions, then we write \(\Subtype{x}{A}{P(x)}\) for the \emph{subtype} of elements of \(A\) satisfying \(P\), i.e., for the type whose elements are pairs \((a, p)\) where \(a : A\) and \(p : P(a)\).

We write \(\LEM\) for the law of excluded middle, which claims that \(P\) or not \(P\) holds for every proposition \(P\), i.e., \(\LEM \defeq \forall(P : \Prop).P+\neg P\).
Since we work constructively, we do not assume \(\LEM\), but will explicitly flag its appearance.
Our main use of \(\LEM\) is to show that other assumptions imply it, and thus have no chance of being constructively provable --- they are constructive taboos.

\subsection{Constructive Ordinals}\label{sec:ord-in-hott}

Cantor's original definition of ordinal numbers could, in modern language, be phrased as \emph{sets equipped with a wellfounded trichotomous relation}~\cite{cantor1883ueber}.
Cantor's formulation of \emph{wellfoundedness} takes the form that \emph{every nonempty subset has a least element}, while \emph{trichotomous} means that we have ${\forall x,y. (x < y) \vee (x = y) \vee (x > y)}$.
In a constructive setting, these notions are not as well-behaved as one would like, and naturally occurring examples often lack trichotomy (see Taylor~\cite{taylor1996intuitionistic} for a further argument).
A set-theoretic constructive formulation due to Powell~\cite{Powell1975} is that of \emph{hereditarily transitive sets} (that is, sets $\alpha$ such that $x \in y \in z \in \ldots \in \alpha$ implies $z \in \alpha$ and $x \in z$).
Grayson~\cite{Grayson1978,Grayson1979} preferred the equivalent formulation of a set with a binary relation that is transitive, extensional, and wellfounded.
This formulation works very well, and was more recently developed in \emph{homotopy type theory}~\cite[\S10]{HoTTBook}; see also Escard\'o's work in Agda~\cite{TypeTopologyOrdinals}.
\emph{Transitivity} means, as always, simply $x < y \to y < z \to x < z$, and \emph{extensionality} means $(\forall z. z < x \leftrightarrow z <y) \to x = y$.
A constructively well-behaved notion of wellfoundedness was given by Aczel~\cite{Aczel1977} as an inductive predicate: an element is \emph{accessible} if all elements below it are accessible, and the order is wellfounded if all elements are accessible.
Wellfoundedness is equivalent to transfinite induction: given a type family~\(P\) over~\(\alpha\),
to prove \(P(x)\) for all \(x : \alpha\), it suffices to prove
\(P(x)\) for all \(x : \alpha\) assuming that \(P(y)\) already holds for \(y < x\), i.e.,
\[\forall(x : \alpha).\,\big(\forall(y : \alpha).\,y < x \to P(y)\big) \to P(x)\] implies \(\forall(x : \alpha).\,P(x)\).

To recap, this paper adopts the following definition.
\begin{definition}[Ordinals \flink{Definition-1}]
  An \emph{ordinal} is a type with a transitive, extensional and wellfounded
  binary relation that is valued in propositions. The relation is typically denoted by \(<\).
\end{definition}
Because of wellfoundedness, the relation is irreflexive, i.e. \(\neg (x < x)\) holds.
As observed by Escard\'o~\cite[{\texttt{Ordinals.Type}}]{TypeTopologyOrdinals}, extensionality implies that the underlying type of an ordinal is a set.

Simple examples of ordinals include the finite types \(\mathbf{n} = \{0 < 1 < \ldots < n - 1\}\), and the infinite ordinal \(\omega\) with underlying type \(\Nat\) and order relation given by the usual order on \(\Nat\).
Any proposition~$P$ can be viewed as an ordinal in a trivial way, where the order is chosen to be constantly empty.
Moreover, the type \(\Omega\) of truth values, i.e., the subtype of the universe containing exactly the propositions, is an ordinal, where \(P < Q\) is defined to mean \(\neg P \wedge Q\).

There are many alternative constructive approaches to ordinals, such as ordinal notation systems~\cite{XNF2020}, Brouwer trees~\cite{eremondi2024}, or wellfounded trees with finite or countable branchings~\cite{MartinLof1970,Coquand2023}, to name a few, but the definition of an ordinal adopted here is the more general one.

The following is a useful observation.
\begin{lemma}[\flink{Lemma-2}]\label{subtype-ordinal}
  Given a type family \(P\) on an ordinal \(\alpha\), the induced relation
  \((a,p) < (a',p') \defeq a < a'\) on \(\Sigma(x : \alpha). \, P(x)\) is
  transitive and wellfounded.
\end{lemma}
\begin{proof}
  Transitivity is immediate. For wellfoundedness, we consider the type family
  \(T\) on \(\alpha\) given by
  \(T(x) \defeq \forall(p : P(x)).\, (x,p) \text{ is accessible}\).
  We must prove \(\forall (x : \alpha).\,T(x)\), so fix \(x : \alpha\) and let
  \(p : P(x)\); we want \((x,p)\) to be accessible. Since \(\alpha\) is
  wellfounded, we may assume that \({\forall(y : \alpha).\,y<x \to T(y)}\) holds.
  But then,
  \(\forall(y : \alpha).\,\forall(q : P(y)).\,\pa*{(y,q) < (x,p) \to (y,q) \text{ is
    accessible}}\) holds, proving the accessibility of \((x,p)\).
\end{proof}
Note that, constructively, the induced relation may fail to be extensional, as
observed by Michael
Shulman~\cite[\texttt{Ordinals.ShulmanTaboo}]{TypeTopologyOrdinals}.

Given an element \(a\) of an ordinal \(\alpha\), the \emph{initial
segment} \(\alpha \initseg a\) determined by \(a\) is the ordinal
with underlying type
\[
  \alpha \initseg a \defeq \Subtype{x}{\alpha}{x < a}
\]
and with order relation inherited from \(\alpha\)~\cite[after Ex~10.3.18]{HoTTBook}.  This
basic construction will prove absolutely fundamental in this paper, and we will
often work with ordinals by characterizing their initial segments, see~\cref{initial-segment-of-sum,initial-segment-of-product,initial-segment-of-abstrexp,initial-segment-of-sup,initial-segment-of-listexp}.
It is based on the fact that an ordinal (or element of an ordinal) is fully characterized by what its predecessors are, by extensionality.
The predecessors of an ordinal $\alpha$ are, by definition, the initial segments of $\alpha$.

\subsection{The Ordinal of Ordinals}\label{subsec:ord-of-ords}
An \emph{ordinal equivalence} from \(\alpha\) to \(\beta\) is given by
an equivalence $f : \alpha \to \beta$ between the underlying types
that is order preserving, and whose inverse $f^{-1} : \beta \to \alpha$
is also order preserving (a property that follows automatically classically, but not constructively).
We note that any ordinal equivalence is also order reflecting.
Thanks to univalence, the identity type \(\alpha = \beta\) is equivalent to the
type of ordinal equivalences from \(\alpha\) to
\(\beta\)~\cite[{\texttt{Ordinals.Equivalence}}]{TypeTopologyOrdinals}, allowing us to
construct identifications between ordinals via equivalences.

The type \(\Ord\) of small ordinals is itself a (large)
ordinal~\cite[Thm.~10.3.20]{HoTTBook} by setting
\[
  \alpha < \beta \defeq \Subtype{b}{\beta}{\alpha = \beta \initseg b},
\]
i.e., \(\alpha\) is strictly smaller than \(\beta\) if \(\alpha\) is an initial segment of~\(\beta\) determined by some (necessarily unique) element \(b : \beta\); it is worth noting that this generalizes the order on \(\Omega\) discussed above.
We also note that proving extensionality for this order uses univalence.

Moreover, \(\Ord\) forms a poset by defining \(\alpha \leq \beta\)
as any of the following equivalent conditions~\cite{TypeTopologyOrdinals} (see
also~\cite[Prop.~9]{dJKNFX2023}):
\begin{enumerate}[label=(\roman*)]
\item\label{poset-non-strict-order} \(\gamma < \alpha\) implies \(\gamma < \beta\) for all small ordinals \(\gamma\);
\item\label{poset-init-seg} for every \(a : \alpha\), there exists a (necessarily
  unique) \(b : \beta\) with \(\alpha \initseg a = \beta \initseg b\);
\item\label{poset-simulation}
  there is a simulation $f : \alpha \to \beta$,
\end{enumerate}
where a \emph{simulation} is an order preserving function such that for all
\(x : \alpha\) and \(y < f(x)\) we have \(x' < x\) with \(f(x') = y\).
Note that \(\leq\) is proposition-valued (in particular,
\ref{poset-simulation} is a proposition: there is at most one simulation
between any two ordinals). Moreover, the relation is antisymmetric: if
\(\alpha \le \beta\) and \(\beta \le \alpha\), then \(\alpha = \beta\), a fact
that we will often use tacitly.
The equivalence between \ref{poset-init-seg} and \ref{poset-simulation} is
due to the first fact in the following lemma.
\begin{lemma}[{\cite[{\texttt{Ordinals.Maps}}]{TypeTopologyOrdinals}} \flink{Lemma-3}]%
  \label{simulation-facts}\leavevmode
  \begin{enumerate}[label=(\roman*)]
  \item\label{simulations-preserve-initial-segments} Simulations preserve initial
    segments: if \(f : \alpha \to \beta\) is a simulation, then
    \(\beta \initseg f \,a \;=\; \alpha \initseg a\).
  \item Simulations are injective and order reflecting.
  \item Surjective simulations are precisely ordinal equivalences.
  \end{enumerate}
\end{lemma}

In a classical metatheory, every order preserving function induces a simulation, so that classically \(\alpha \leq \beta\) holds if and only if there exists an order preserving function \(\alpha \to \beta\).
However, as we will see in \cref{sec:taboos}, this is not true constructively.
Compared to mere order preserving maps, simulations between ordinals are rather
well behaved (as witnessed by \cref{simulation-facts}).

An element \(\bot\) of an ordinal \(\alpha\) is \emph{least} if \(\bot \preceq a\) for all \(a\) in \(\alpha\), where \(x \preceq y\) holds if \(u < x\) implies \(u < y\) for all \(u : \alpha\).
Note that for the ordinal of ordinals \(\Ord\), the order \(\alpha \preceq \beta\) coincides with the order \(\alpha \leq \beta\) via characterization \ref{poset-non-strict-order} above.
We remark that \(\bot\) is the least element of \(\alpha\) if and only if there
are no elements \(x\) in \(\alpha\) with \(x < \bot\), equivalently if and only
if the map \(\star \mapsto \bot : \One \to \alpha\) is a simulation.

\subsection{Suprema}
The poset of (small) ordinals has suprema (least upper bounds) of arbitrary
families indexed by small types~\cite[Thm.~5.8]{deJongEscardo2023}.
Given a family of ordinals \(F_{\sbullet} : I \to \Ord\), its supremum
\(\sup F_{\sbullet}\) can be constructed as the total space
\(\Sigma (i:I). \, F_i\), quotiented by the relation that identifies \((i,x)\)
and \((j,y)\) if \(F_i \initseg x = F_j \initseg y\), see
also~\cite[\S2.5]{Grayson1982}.
The order relation is as follows: the equivalence class of \((i,x)\) is strictly
below that of \((j,y)\) if \(F_i \initseg x < F_j \initseg y\).
As shown in~\cite[Thm.~5.8]{deJongEscardo2023}, we have a simulation
\([i,-] : F_i \leq \sup F_{\sbullet}\) for each \(i : I\).
Moreover, these maps are jointly surjective and can be used to characterize
initial segments of the supremum~\cite[Lem.~15]{dJKNFX2023}: for every
\(y : \sup F_{\sbullet}\) there exists some \(i : I\) and \(x : F_i\)
such that
\begin{fequation}{Eq-1}\label{initial-segment-of-sup}
  y = [i,x] \text{ and } (\sup F_{\sbullet}) \initseg y = F_i \initseg x.
\end{fequation}
We stress that the existence of the pair \((i,x)\) is expressed using
\(\exists\) (i.e., the propositional truncation of \(\Sigma\)).
We will only use \cref{initial-segment-of-sup} to prove propositions, so that we
can use the universal property of the truncation to obtain an
actual pair.

We write \(\alpha \vee \beta\) for the binary join of \(\alpha\) and \(\beta\),
i.e., for the supremum $\sup F_{\sbullet}$ of the two-element family
\(F_{\sbullet} : \Two \to \Ord\) with \(F_0 = \alpha\) and \(F_1 =
\beta\). Throughout the rest of the paper, we will moreover find it convenient
to represent \(\alpha \vee \sup_{i : I} F_{i}\) as \(\sup F'_{\sbullet}\), where
\(F'_{\sbullet} : {\One + I} \to \Ord\) is given by
\(\inl \star \mapsto \alpha\) and \(\inr i \mapsto F_i\).
\begin{remark}[\flink{Remark-4}]\label{inhabited-sup}
  A particular case of the above is where \(\alpha \defeq \Zero\) in which case
  \(\sup F'_{\sbullet}\) equals \(\sup F_{\sbullet}\).
  The upshot of using \(F'\) is that its indexing type is always inhabited,
  which, constructively we cannot decide for a general type \(I\).
\end{remark}

An operation on ordinals \(F : \Ord \to \Ord\) is \emph{continuous} if it preserves suprema, i.e., if \(F(\sup G_{\sbullet}) = \sup ({F} \circ {G_{\sbullet}})\) for every \(G_{\sbullet} : I \to \Ord\). Continuous operations are automatically monotone, since \(\alpha \leq \beta\) if and only if \(\alpha \vee \beta = \beta\). Thus if \(F\) is continuous and \(\alpha \leq \beta\), we have \(F(\alpha) \vee F(\beta) = F(\alpha \vee \beta) = F(\beta)\) and hence \(F(\alpha) \leq F(\beta)\).

\subsection{Addition and Multiplication}
We start by reviewing the construction of addition and multiplication of
ordinals via coproducts and cartesian products, and compare with
their classical definitions via cases.

As is well known (cf.~\cite[{\texttt{Ordinals.Arithmetic}}]{TypeTopologyOrdinals} and
\cite[Thm.~61]{KNFX2023}), addition of two ordinals is given by taking the
coproduct of the underlying types, keeping the original order in each component,
and additionally requiring that all elements in the left component are smaller
than anything in the right component.
Initial segments of a sum of ordinals can be calculated as follows:
\begin{fequation}{Eq-2}\label{initial-segment-of-sum}
  \begin{split}
  &{(\alpha + \beta)} \initseg \inl a \;=\; \alpha \initseg a; \\
  &{(\alpha + \beta)} \initseg \inr b \;=\; \alpha + (\beta \initseg b).
  \end{split}
\end{fequation}

Multiplication of two ordinals is given by equipping the cartesian product of
the underlying types with the reverse lexicographic order, i.e.,
\((a' , b') < (a , b)\) holds if either \(b' < b\), or \(b' = b\) and \(a' < a\).
Initial segments of products of ordinals can be calculated as follows:
\begin{fequation}{Eq-3}\label{initial-segment-of-product}
  {(\alpha \times \beta)} \initseg (a , b) \;=\;
  \alpha \times (\beta \initseg b) + (\alpha \initseg a).
\end{fequation}
If \(b' < b\), the ordinal equivalence witnessing
\eqref{initial-segment-of-product} sends \((a' , b')\) to \(\inl(a' , b')\), and
if \(b' = b\) and \(a' < a\), then \((a' , b')\) is sent to \(\inr a'\).
Note that neither addition nor multiplication is commutative, since e.g., \(\omega + \One \neq \One + \omega\) and \(\omega \times \Two \neq \Two \times \omega\).

\begin{lemma}[\flink{Lemma-5}]\label{+-x-right-monotone}
  Addition and multiplication are monotone in their right arguments, i.e.,
  \(\beta \le \gamma\) implies \({\alpha + \beta} \le {\alpha + \gamma}\) and
  \({\alpha \times \beta} \le {\alpha \times \gamma}\).
\end{lemma}
\begin{proof}
  Given a simulation \(f : \beta \le \gamma\), we get
  \({\alpha + \beta} \le {\alpha + \gamma}\) via
  \(\inl a \mapsto \inl a\) and \(\inr b \mapsto \inr(f\,b)\) as well as
  \({\alpha \times \beta} \le {\alpha \times \gamma}\) by mapping \((a,b)\) to
  \((a,f\,b)\).
\end{proof}

In the rest of this section, we study how addition and multiplication interact with suprema, and revisit the classical definitions of these operations by cases.

\subsubsection{Defining equations for ordinal multiplication}\label{sec:eqs-for-multiplication}

Let us now reconsider the specification \eqref{eq:addition-multiplication-spec}
of ordinal multiplication from the introduction. It classically defines
multiplication via the equations \(\alpha \times 0 = 0\),
\(\alpha \times (\beta + 1) = (\alpha \times \beta) + \alpha\) and
\({\alpha \times \lambda} = \sup_{\beta < \lambda}{(\alpha \times \beta)}\) for
limit ordinals \(\lambda\).
We first note that ordinal multiplication defined using cartesian products actually satisfies the limit ordinal clause of the specification for arbitrary suprema:

\begin{lemma}[\flink{Lemma-6}]\label{multiplication-is-continuous}
  Ordinal multiplication is continuous in its right argument:
  \(\alpha \times \sup F_{\sbullet} = \sup(\alpha \times F_{\sbullet})\).
\end{lemma}
\begin{proof}
  By the least upper bound property of \(\sup(\alpha \times F_{\sbullet})\) and
  \cref{+-x-right-monotone}, it suffices to show that
  \({\alpha \times \sup F_{\sbullet}} \le \sup(\alpha \times F_{\sbullet})\).
  Given \(a : \alpha\) and \(y : \sup F_{\sbullet}\), we use
  \cref{initial-segment-of-product,initial-segment-of-sup} to get the existence
  of \(i : I\) and \(x : F_i\) with
  \begin{align*}
    (\alpha \times \sup F_{\sbullet}) \initseg (a,y)
    &= \alpha \times (\sup F_{\sbullet} \initseg y) + (\alpha \initseg a)\\
    &= \alpha \times (F_i \initseg x) + (\alpha \initseg a)\\
    &= \sup (\alpha \times F_{\sbullet}) \initseg [i,(a,x)]
  \end{align*}
  hence \(\alpha \times \sup F_{\sbullet} \le \sup(\alpha\times F_{\sbullet})\).
\end{proof}

The specification \eqref{eq:addition-multiplication-spec} is based on the classical characterization of ordinals as
either successors or limits, where successor ordinals are of the form
\(\beta = \gamma + 1\) and limit ordinals are of the form
\(\lambda = \sup_{\gamma < \lambda}\,\gamma\).
This classification is not available constructively, but a weaker variant of it
is, where both cases are combined: every ordinal is (constructively) the
supremum of the successors of its predecessors:
\begin{lemma}[{\cite[\S2.5]{Grayson1982}, \cite[Ex.~5.21]{deJongEscardo2025}} \flink{Lemma-7}]%
  \label{ordinal-as-sup-of-successors}
  Every ordinal \(\beta\) satisfies the equation
  \(\beta = \sup_{b:\beta}(\beta \initseg b + \One)\).
\end{lemma}
\begin{proof}
  For any \(b : \beta\) we have \(\beta \initseg b + \One \le \beta\), so
  \(\beta\) is an upper bound for the family. To see that it is the least,
  suppose that we have simulations
  \(g_b : {\beta \initseg b + \One} \le \gamma\) for all \(b : \beta\).
  Thus, for each \(b : \beta\), we have \(c_b : \gamma\) such that
  \({\beta \initseg b} = {(\beta \initseg b + \One)} \initseg \inr \star = {\gamma
    \initseg c_b}\).
  Hence, the map \(b \mapsto g_b(\inr \star)\) is a simulation
  \(\beta \le \gamma\).
\end{proof}

Thanks to \cref{ordinal-as-sup-of-successors}, and by generalizing the equation
\(\alpha \times \lambda = \sup_{\beta < \lambda}{(\alpha \times \beta)}\) from
limit ordinals to arbitrary suprema, we obtain a specification that uniquely
characterizes ordinal multiplication, also constructively.

\begin{proposition}[\flink{Proposition-8}]\label{multiplication-defining-equations}
  Ordinal multiplication satisfies the following equations
  \[
    \alpha \times \Zero = \Zero, \quad
    \alpha \times (\beta + \One) = (\alpha \times \beta) + \alpha, \quad
    \alpha \times \sup F_{\sbullet} = \sup (\alpha \times F_{\sbullet}).
  \]

  In turn, these equations uniquely characterize multiplication via the recursive equation
  \[
    \alpha \times \beta = \sup_{b : \beta} (\alpha \times (\beta \initseg b) + \alpha)
  \]
  which may act as a definition via transfinite induction.
\end{proposition}

The last expression involves \(\alpha \times ({\beta \initseg b})\), where the
right factor \(\beta \initseg b\) is strictly smaller than \(\beta\), so that
the recursive equation indeed leads to a valid definition by transfinite
induction in \(\Ord\) on \(\beta\).

\begin{proof}
  It is easy to construct an ordinal equivalence between \({\alpha \times \Zero}\) and
  \(\Zero\). For the second equation, an ordinal equivalence is given by the mapping
  \((a,\inl b) \mapsto \inl(a,b)\) and \((a,\inr\star) \mapsto \inr a\).
  The supremum equation is \cref{multiplication-is-continuous}.
  Given the three equations, we calculate that
  \begin{align*}
    \alpha \times \beta &= \alpha \times {\sup_{b : \beta}\big((\beta \initseg b) + \One\big)}
    &&\text{(by \cref{ordinal-as-sup-of-successors})} \\
    &= \sup_{b : \beta} (\alpha \times \big((\beta \initseg b) + \One\big))
    &&\text{(by the equation for suprema)} \\
    &= \sup_{b : \beta} (\alpha \times (\beta \initseg b) + \alpha),
    &&\text{(by the equation for successors)}
  \end{align*}
  as desired.
\end{proof}

Notice that the above proof does not refer to the equation for \(\Zero\). In
fact, the equation \(\alpha \times \Zero = \Zero\) is redundant as it follows
from the equation for suprema by considering the empty family.
We also point out that the recursive equation was already observed by
Grayson~\cite[\S3.1]{Grayson1982} (without proof).

\subsubsection{Defining equations for ordinal addition}\label{sec:eqs-for-addition}
For ordinal addition, the situation is slightly more involved.
First of all, ordinal addition is not continuous: the equation
\({\alpha + \sup F_{\sbullet}} = \sup {(\alpha + F_{\sbullet})}\) fails for the
empty family whenever \(\alpha \ne \Zero\), since \(\alpha + \Zero = \alpha\).
However, it is still the case that \({\alpha + (-)}\) preserves suprema for any
family whose index type is inhabited, i.e., we have
\({\alpha + \sup_{i : I}F_i} = \sup_{i : I}{(\alpha + F_i)}\) whenever \(I\) is
inhabited.
Here, we recall that a type is \emph{inhabited} if we have an element of its
propositional truncation, which is in general stronger than the classically
equivalent statement that the type is nonempty.

The preservation of inhabited suprema and the equation for zero are together
equivalent to preserving suprema up to a binary join with a fixed ordinal. We have
\begin{gather*}
\begin{alignat*}{2}
  \alpha + \Zero = \alpha \quad\text{ and }\quad
  {\alpha + \sup_{i : I}F_i} &= \sup_{i : I}{(\alpha + F_i)}
  &&\quad\text{ for inhabited \(I\)} \\
  \shortintertext{\centering\(\iff\)}
  {\alpha + \sup_{i : I} F_{i}} &= {\alpha \vee \sup_{i : I} {(\alpha + F_{i})}}
  &&\quad\text{ for arbitrary \(I\)}.
\end{alignat*}
\end{gather*}

Indeed, given the top equations, thanks to and with notation
from~\cref{inhabited-sup}, we have
\[
  \alpha + \sup F_{\sbullet} = \alpha + \sup F'_{\sbullet} = \sup (\alpha +
  F'_{\sbullet}) = \alpha \vee \sup(\alpha + F_{\sbullet}),
\]
where we used that the indexing type of \(F'\) is inhabited.
Conversely, given the bottom equation, we recover \({\alpha + \Zero} = \alpha\)
by taking the empty family.
For a family \(F\) indexed by an inhabited type~\(I\), we get
\({\alpha + \sup_{i : I}F_i} = \sup_{i : I}{(\alpha + F_i)}\) because, for
inhabited \(I\) we have \(\sup_{i : I}{(\alpha + F_i)} \ge \alpha\) since
\(\alpha = {\alpha + \Zero} \le \alpha + {F_i}\) holds for any \(i : I\) by
monotonicity.

We thus consider a generalized continuity property where an operation
preserves suprema up to a binary join with a fixed ordinal.

\begin{lemma}[\flink{Lemma-9}]\label{+-preserves-suprema}
  For every ordinal \(\alpha\) and family \(F\) of ordinals, we have
  \[
    \alpha + \sup {F_{\sbullet}} = \alpha \vee {\sup (\alpha + F_{\sbullet})}.
  \]
\end{lemma}
\begin{proof}
  We use antisymmetry. In general, we have \(\alpha \le {\alpha + \beta}\) and
  hence \(\alpha \le {\alpha + \sup F_{\sbullet}}\) holds.
  Moreover, we have \(\alpha + F_i \le {\alpha + \sup {F_{\sbullet}}}\) as
  addition is monotone in its right argument. This establishes
  \(\alpha \vee {\sup (\alpha + F_{\sbullet})} \le \alpha + \sup
  {F_{\sbullet}}\).

  For the reverse inequality we use that
  \(\alpha \vee {\sup (\alpha + F_{\sbullet})}\) is equal to
  \(\sup_{j : \One + I} G_j\) where \(G\) is given by
  \(G(\inl\star) \defeq \alpha\) and \(G(\inr i) \defeq \alpha + F_i\).
  Now consider \(z : \alpha + \sup F_{\sbullet}\); we need to find a
  (necessarily unique) \(z' : \sup G_{\sbullet}\) such that the
  initial segment determined by \(z\) is equal to the initial segment determined
  by \(z'\).
  In case \(z = \inl a\), we calculate that
  \begin{align*}
    {(\alpha + \sup F_{\sbullet})} \initseg \inl a
    &= {\alpha \initseg a} &&\text{(by \cref{initial-segment-of-sum})} \\
    &= \sup G_{\sbullet} \initseg [\inl \star,a]
      &&\text{(by \cref{initial-segment-of-sup})}.\\
  %
  \intertext{If \(z = \inr y\), there exists \(i : I\) and \(x : F_i\) such that
  \(y = [i,x]\) and we calculate that}
    ({\alpha + \sup F_{\sbullet}}) \initseg \inr {[i,x]}
    &= \alpha + ({\sup F_{\sbullet} \initseg [i,x]})
    &&\text{(by \cref{initial-segment-of-sum})} \\
    &= \alpha + ({F_i \initseg x})
    &&\text{(by \cref{initial-segment-of-sup})} \\
    &= {(\alpha + F_i)} \initseg {\inr x}
    &&\text{(by \cref{initial-segment-of-sum})} \\
    &= {\sup G_{\sbullet}} \initseg {[\inr i,\inr x]}
    &&\text{(by \cref{initial-segment-of-sup})}.
  \end{align*}
  Hence, we have a simulation from \(\alpha + \sup F_{\sbullet}\) to
  \(\alpha \vee \sup(\alpha + F_{\sbullet})\).
\end{proof}

With \cref{+-preserves-suprema}, we can prove the analogue of
\cref{multiplication-defining-equations} for addition of ordinals, where again the equation for zero is redundant, and
the recursive equation was already observed (without proof) by
Grayson~\cite[\S3.1]{Grayson1982}.

\begin{proposition}[\flink{Proposition-10}]\label{addition-defining-equations}
  Ordinal addition satisfies the following equations
  \[
    \alpha + \Zero = \alpha, \quad
    \alpha + (\beta + \One) = (\alpha + \beta) + \One, \quad
    \alpha + \sup F_{\sbullet} = \alpha \vee \sup (\alpha + F_{\sbullet}).
  \]

  In turn, these equations uniquely characterize addition via the recursive equation
  \[
    \alpha + \beta = \alpha \vee \sup_{b : \beta} (\alpha + (\beta \initseg b) + \One)
  \]
  which may act as a definition via transfinite induction.
\end{proposition}
\begin{proof}
  It is easy to construct an ordinal equivalence between \(\alpha + \Zero\) and
  \(\alpha\), as well as between \(\alpha + (\beta + \gamma)\) and
  \((\alpha + \beta) + \gamma\), which generalizes the second equation.
  The supremum equation is \cref{+-preserves-suprema}.
  Given the three equations, we use \cref{ordinal-as-sup-of-successors} to
  calculate that
  \[
    \alpha + \beta = \alpha + {\sup_{b : \beta}(\beta \initseg b + \One)}
    = \alpha \vee \sup_{b : \beta} (\alpha + (\beta \initseg b + \One))
    = \alpha \vee \sup_{b : \beta} ((\alpha + (\beta \initseg b)) + \One).
    \qedhere
  \]
\end{proof}

\subsection{Expectations on Exponentiation}

We can now state and make precise the specification \eqref{eq:intro-spec} in the language of homotopy type theory.
The following equations classically define ordinal exponentiation.
\begin{fequation}{Eq-double-dagger}\label{exp-full-spec}
  \begin{aligned}
    \alpha^0 & = \One \\
    \alpha^{\beta+\One} & = \alpha^\beta \times \alpha \\
    \alpha^{\sup_{i:I} F_i} & = \sup_{i:I} \pa*{\alpha^{F_i}}
                        && \text{(if $\alpha \not= 0$ and $I$ inhabited)} \\
    \Zero^\beta & = \Zero && \text{(if \(\beta \neq \Zero\))}
  \end{aligned}
  \tag{\(\ddagger\)}
\end{fequation}

The final clause \(\Zero^\beta = \Zero\) has the side condition \(\beta \neq \Zero\) in order to not clash with \(\alpha^0 = \One\).
Classically, the supremum clause is equivalent to the usual equation involving limit ordinals, but formulating it for arbitrary inhabited suprema has the advantage of ensuring that exponentiation directly preserves these suprema.
\begin{lemma}[\flink{Lemma-11}]\label{basic-facts-from-spec}
  The zero and successor clauses in the specification together imply the equations \(\alpha^\One = \alpha\) and
  \(\alpha^{\Two} = \alpha \times \alpha\).
  Moreover, the supremum clause in the specification implies that
  \(\alpha^{(-)}\) is monotone for \(\alpha\neq 0\), i.e., if \(\beta \leq \gamma\) then \(\alpha^\beta \leq \alpha^\gamma\) for \(\alpha\neq 0\). \qedNoProof
\end{lemma}

The bad news is that constructively there cannot be an operation satisfying all
of the equations in~\eqref{exp-full-spec}:

\newcommand{\expop}{\mathrm{exp}}
\begin{proposition}[\flink{Proposition-12}]\label{thm:no-go-exp}
  There is an exponentiation operation \[\expop : \Ord \times \Ord \to \Ord\]
  satisfying the specification~\eqref{exp-full-spec} if and only if\/ \(\LEM\) holds.
\end{proposition}
\begin{proof}
  Using \(\LEM\), such an operation can be defined by cases.
  Conversely, suppose we had such an operation \(\expop\) and let \(P\) be an arbitrary
  proposition.
  As explained in \cref{sec:ord-in-hott}, any proposition can be viewed as an
  ordinal, and we consider the sum of two such ordinals:
  \(\alpha \defeq {P + \One}\).
  Obviously \(\alpha \neq \Zero\) and \(\Zero \leq \One\), so \cref{basic-facts-from-spec} yields
  \(\One = \expop (\alpha , \Zero) \leq \expop (\alpha , \One) = \alpha\).
  Hence, we get a simulation \(f : \One \to \alpha\).
  Now, either \(f\,\star = \inl p\), in which case \(P\) holds, or
  \(f\,\star = \inr \star\).
  In the latter case, we can prove \(\lnot P\):
  simulations preserve least elements, so assuming \(p : P\) we must have
  \(f\,\star = \inl p\), which contradicts \(f\,\star = \inr \star\).
\end{proof}

The good news is that, if we assume that the base is positive, we \emph{can} define a well behaved ordinal exponentiation
operation \(\alpha^\beta\) that satisfies the specification \eqref{exp-full-spec}.
Assuming \(\alpha \ge \One\), it is convenient, and completely analogous to what
we did in~\cref{sec:eqs-for-addition}, to consider the \emph{equivalent}
specification of
exponentiation which combines the zero and supremum cases and which drops the
requirement that \(I\) is inhabited:
\begin{fequation}{Eq-double-dagger'}\label{exp-strong-spec}
  \begin{aligned}
    \alpha^{\beta+\One} & = \alpha^\beta \times \alpha \\
    \alpha^{\sup_{i:I} F_i} & = \One \vee \sup_{i:I} \pa*{\alpha^{F_i}}
  \end{aligned}
  \tag{\(\ddagger'\)}
\end{fequation}

\section{Abstract Algebraic Exponentiation}\label{sec:abstract-approach}

Let us start by presenting a definition of exponentiation that is fully guided by the equations we expect or want.
We know from \cref{thm:no-go-exp} that we cannot hope to define exponentiation \(\alpha^\beta\) for arbitrary \(\alpha\) and \(\beta\), so in order to avoid the case distinction on whether $\alpha = 0$ or not, let us restrict our attention to the case when $\alpha \geq \One$.
This means that it is sufficient to define an exponentiation operation which satisfies the specification \eqref{exp-strong-spec}.
Just like in \cref{sec:eqs-for-multiplication}, these equations
\emph{uniquely determine} \(\alpha^\beta\) for \(\alpha \geq \One\) via
transfinite recursion, thanks to \cref{ordinal-as-sup-of-successors}.
Indeed, we have
\begin{equation*}
  \begin{aligned}
    \alpha^\beta
    &= \alpha^{\sup_{b:\beta}({\beta \initseg b} \,+\, \One)} &&\text{(by \cref{ordinal-as-sup-of-successors})}\\
    &= \One \vee \sup_{b:\beta}(\alpha^{{\beta \initseg b} \,+\, \One}) &&\text{(to satisfy \eqref{exp-strong-spec}, supremum case)} \\
    &= \One \vee \sup_{b:\beta}(\alpha^{\beta \initseg b} \times \alpha) &&\text{(to satisfy \eqref{exp-strong-spec}, successor case).}
  \end{aligned}
\end{equation*}
Since it will be convenient to work with a single supremum instead, we adopt the
following definition.

\begin{definition}[Abstract exponentiation, $\abstrexp{\alpha}{\beta}$ \flink{Definition-13}]
  For a given ordinal $\alpha$, we define the operation
  $\abstrexp{\alpha}{(-)} : \Ord \to \Ord$ by transfinite recursion as follows:
  \begin{equation*}
    \alpha^\beta
    \defeq \sup_{\One + \beta}(\inl \star \mapsto \One; \inr b \mapsto \alpha^{\beta \initseg b} \times \alpha).
  \end{equation*}
\end{definition}

The definition of exponentiation \(\abstrexp{\alpha}{\beta}\) does not rely on \(\alpha\) being positive, but most properties of \(\abstrexp{\alpha}{\beta}\) will require it.
Note that exponentiation itself is always positive, i.e.,
$\abstrexp{\alpha}{\beta} \geq \One$ holds by construction, making it possible
to iterate well behaved exponentiation.
We will write \(\bot \defeq [\inl \star,\star]\) for the least element of an
exponential \(\abstrexp{\alpha}{\beta}\).
Using \cref{initial-segment-of-sup,initial-segment-of-product} we get the following characterization of initial segments.

\begin{proposition}[Initial segments of \(\abstrexp{\alpha}{\beta}\) \flink{Proposition-14}]%
  \label{initial-segment-of-abstrexp}
  For \(a : \alpha\), \(b : \beta\) and
  \(e : \abstrexp{\alpha}{\beta \initseg b}\), we have
  \[
    {\abstrexp{\alpha}{\beta} \initseg [\inr b,(e,a)]}
    \;
    =
    \;
    {\abstrexp{\alpha}{\beta \initseg b} \times (\alpha \initseg a) +
      \abstrexp{\alpha}{\beta \initseg b} \initseg e}.\eqno\qed
  \]
\end{proposition}

We can put this characterization to work immediately, to prove that abstract exponentiation is monotone in the exponent for both the weak and the strict order of ordinals.

\begin{proposition}[\flink{Proposition-15}]\label{abstrexp-monotone}
  Abstract exponentiation is monotone in the exponent:
  if \(\beta \leq \gamma\) then \(\alpha^\beta \leq \alpha^\gamma\). Furthermore, if \(\alpha > \One\),
  then it moreover preserves the strict order, i.e., if \(\alpha > \One\) and \(\beta < \gamma\), then \(\alpha^\beta < \alpha^\gamma\).
\end{proposition}
\begin{proof}
  Let \(f : \beta \to \gamma\) be a simulation. By \cref{simulation-facts} we
  have
  \begin{align*}
    \alpha^\beta &\equiv
    \sup_{\One + \beta}(\inl \star \mapsto \One; \inr b \mapsto \alpha^{\beta \initseg b} \times \alpha) \\
    &= \sup_{\One + \beta}(\inl \star \mapsto \One; \inr b \mapsto \alpha^{\gamma \initseg {f b}} \times \alpha) \\
    &\le \sup_{\One + \gamma}(\inl \star \mapsto \One; \inr c \mapsto \alpha^{\gamma \initseg c} \times \alpha) \equiv \alpha^\gamma,
  \end{align*}
  where the inequality holds because the supremum gives the least upper bound.

  For the second claim, suppose we have \(\alpha > \One\) and \(\beta < \gamma\), i.e., we have \(\One = \alpha \initseg a_1\) and \(\beta = \gamma \initseg c\) for some \(a_1 : \alpha\) and \(c : \gamma\).
  Then, using \cref{initial-segment-of-abstrexp} and the least element
  \(\bot\) of \(\abstrexp{\alpha}{\gamma \initseg c}\), we calculate that
  \begin{align*}
    \abstrexp{\alpha}{\gamma} \initseg [\inr c , (\bot , a_1)]
    &= \abstrexp{\alpha}{\gamma \initseg c} \times (\alpha \initseg a_1)
      + \abstrexp{\alpha}{\gamma \initseg c} \initseg \bot \\
    &= \abstrexp{\alpha}{\gamma \initseg c} \times \One + \Zero
    = \abstrexp{\alpha}{\gamma \initseg c} = \abstrexp{\alpha}{\beta}.
  \end{align*}
  Hence, \(\alpha^\beta < \alpha^\gamma\), as desired.
\end{proof}

Using monotonicity, we can now prove that abstract exponentiation is well behaved whenever the base is positive.

\begin{theorem}[\flink{Theorem-16}]\label{abstract-exp-satisfies-spec}
  Assuming \(\alpha \ge \One\), abstract exponentiation \(\abstrexp{\alpha}{\beta}\) satisfies the specification \eqref{exp-strong-spec} \textup{(}and hence also the specification \eqref{exp-full-spec}\textup{)}.
\end{theorem}
\begin{proof}
  For the successor clause, we want to show
  \(\abstrexp{\alpha}{\beta + \One} = \abstrexp{\alpha}{\beta}\times\alpha\). First note that if \(\alpha \geq \One\), then \(\abstrexp{\alpha}{\One} = \alpha\), by the definition of \(\abstrexp{\alpha}{\One}\) as a supremum. The result then follows from the more general statement in~\cref{abstrexp-by-+} below.\footnote{In the formalization we additionally include a direct proof which is more general in terms of universe levels.}
  Next, for the supremum clause of \eqref{exp-strong-spec},
we want to show \(\abstrexp{\alpha}{\sup F_{{\sbullet}}} = \One \vee \sup {\abstrexp{\alpha}{F_{{\sbullet}}}}\) for a given \(F : I \to \Ord\).
  Since \(F_i \le \sup F_{\sbullet}\) for all \(i : I\), we have
  \(\sup {\abstrexp{\alpha}{F_{{\sbullet}}}} \le \abstrexp{\alpha}{\sup F_{{\sbullet}}}\)
  via \cref{abstrexp-monotone}, and hence \(\One \vee \sup {\abstrexp{\alpha}{F_{{\sbullet}}}} \le \abstrexp{\alpha}{\sup F_{{\sbullet}}}\) since \(\One \leq \abstrexp{\alpha}{\beta}\) for any \(\beta\).
  For the reverse inequality, it suffices to prove
  \(\One \le \One \vee \sup \abstrexp{\alpha}{F_{{\sbullet}}}\) and
  \[
    \abstrexp{\alpha}{\sup F_{\sbullet} \initseg y} \times \alpha
    \le \One \vee \sup \abstrexp{\alpha}{F_{{\sbullet}}}
  \]
  for all \(y : \sup F_{\sbullet}\).
  The former is immediate, and for the latter, we note that \cref{initial-segment-of-sup} implies the
  existence of \(i : I\) and \(x : F_i\) such that
  \[
    {\abstrexp{\alpha}{\sup F_{\sbullet} \initseg y} \times \alpha}
    = {\abstrexp{\alpha}{F_i \initseg x} \times \alpha}
    \le {\abstrexp{\alpha}{F_{i}}} \le {\One \vee \sup \abstrexp{\alpha}{F_{{\sbullet}}}}.
  \]
    Finally, since $\alpha \geq \One$ by assumption, \eqref{exp-full-spec} is equivalent to \eqref{exp-strong-spec}.
\end{proof}

The following two propositions establish the expected connections of exponentiation
with addition and multiplication, respectively.
Note that they hold even without the assumption that \(\alpha \geq \One\).

\begin{proposition}[\flink{Proposition-17}]\label{abstrexp-by-+}
  For ordinals $\alpha$, $\beta$ and $\gamma$, we have
  \[
    \alpha^{\beta + \gamma} = \alpha^\beta \times \alpha^\gamma.
  \]
\end{proposition}
\begin{proof}
  We do transfinite induction on $\gamma$.  Our first observation is that
  \[
    {\alpha^\beta \times \alpha^\gamma} =
    {{\alpha^\beta} \vee {\sup_{c : \gamma}(\alpha^\beta \times \alpha^{\gamma \initseg c} \times \alpha)}},
  \]
  which follows from the fact that multiplication is associative as well as
  continuous on the right (\cref{multiplication-is-continuous}), noting that
  $\vee$ is implemented as a supremum.

  Applying the induction hypothesis, we can rewrite
  $\alpha^\beta \times \alpha^{\gamma \initseg c}$ to
  $\alpha^{\beta + \gamma \initseg c}$, which is
  $\alpha^{(\beta + \gamma) \initseg \inr c}$.  The remaining goal thus is
  \[
    \alpha^{\beta + \gamma} =
    \alpha^\beta \vee \sup_{c : \gamma}\pa{\alpha^{(\beta + \gamma) \initseg \inr c} \times \alpha},
  \]
  which one gets by unfolding the definition on the left and applying antisymmetry of \(\le\).
\end{proof}

\begin{proposition}[\flink{Proposition-18}]\label{abstrexp-by-product}
  For ordinals $\alpha$, $\beta$ and $\gamma$, iterated exponentiation
  can be calculated as follows:
  \[
    \abstrexp{\pa*{\abstrexp{\alpha}{\beta}}}{\gamma} = \abstrexp{\alpha}{\beta \times \gamma}.
  \]
\end{proposition}
\begin{proof}
  We proceed by transfinite induction on \(\gamma\) and use antisymmetry.
  Since exponentials are least upper bounds, and always positive, in one direction it suffices to prove
  \(
    \abstrexp{\pa*{\abstrexp{\alpha}{\beta}}}{\gamma \initseg c} \times \abstrexp{\alpha}{\beta}
      \le \abstrexp{\alpha}{\beta \times \gamma}
  \)
  for all \(c : \gamma\).
  To this end, notice that
  \begin{align*}
    \abstrexp{\pa*{\abstrexp{\alpha}{\beta}}}{\gamma \initseg c} \times \abstrexp{\alpha}{\beta}
    &=\abstrexp{\alpha}{\beta \times {\gamma \initseg c}} \times \abstrexp{\alpha}{\beta}
    &&{\text{(by IH)}} \\
    &=\abstrexp{\alpha}{\beta \times {\gamma \initseg c} \,+\, \beta}
    &&{\text{(by \cref{abstrexp-by-+})}} \\
    &=\abstrexp{\alpha}{\beta \times \pa*{{\gamma \initseg c} + \One}}
    &&\text{(since \(\times\) distributes over \(+\))} \\
    &\le \abstrexp{\alpha}{\beta \times \gamma},
  \end{align*}
  where the final inequality holds because we have \(\gamma \initseg c \;+\; \One \le \gamma\)
  (by \cref{ordinal-as-sup-of-successors}), exponentiation is monotone in the
  exponent (\cref{abstrexp-monotone}), and multiplication is monotone on the
  right (\cref{+-x-right-monotone}).

  For the other inequality, we show
  \(
  \abstrexp{\alpha}{\pa*{\beta \times \gamma} \initseg (b,c)} \times \alpha
  \le \abstrexp{\pa*{\abstrexp{\alpha}{\beta}}}{\gamma}
  \)
  for all \(b : \beta\) and \(c : \gamma\). Indeed we have
  \begin{align*}
    &\abstrexp{\alpha}{\pa*{\beta \times \gamma} \initseg (b,c)} \times \alpha \\
    &= \abstrexp{\alpha}{{\beta \times (\gamma \initseg c) + \beta \initseg b}} \times \alpha
    &&\text{(by \cref{initial-segment-of-product})} \\
    &= \abstrexp{\alpha}{{\beta \times (\gamma \initseg c)}} \times \abstrexp{\alpha}{\beta \initseg b} \times \alpha
    &&\text{(by \cref{abstrexp-by-+})} \\
    &= \abstrexp{\pa*{\abstrexp{\alpha}{\beta}}}{\gamma \initseg c} \times \abstrexp{\alpha}{\beta \initseg b} \times \alpha
    &&\text{(by IH)} \\
    &\le \abstrexp{\pa*{\abstrexp{\alpha}{\beta}}}{\gamma \initseg c} \times \abstrexp{\alpha}{\beta}
    &&\text{(assoc.\ and monotonicity of \({\times}\))} \\
    &\le \abstrexp{\pa*{\abstrexp{\alpha}{\beta}}}{\gamma}.&&\hspace{-1em}\qedhere
  \end{align*}
\end{proof}

While it is quite clear that addition and multiplication of ordinals preserve
decidable equality, it is not obvious at all that exponentiation also preserves
this property --- exponentiation is defined as a supremum, which is defined as a quotient, and it is not the case that quotients preserve decidable equality.
Luckily, the construction introduced in the next section will make this fact
obvious, at least when \(\alpha\) has a trichotomous least element.

\section{Decreasing Lists: a Constructive Formulation of \texorpdfstring{Sierpi\'nski's}{Sierpinski's} Definition}\label{sec:concrete-approach}

As discussed in the introduction, in a classical metatheory, there is a ``non-axiomatic'' construction of exponentials $\alpha^\beta$ for $\alpha \geq \One$, based on functions $\beta \to \alpha$ with finite support~\cite[\S XIV.15]{sierpinski}.
Recall that $\alpha \geq \One$ means that $\alpha$ has a least element $\bot : \alpha$, and that a function $\beta \to \alpha$ has finite support if it is zero almost everywhere, i.e., if it differs from the least element $\bot$ for only finitely many inputs.
Using classical logic, the set of functions $\beta \to \alpha$ with finite support can then be shown to be an ordinal.

Unfortunately, this construction depends on classical principles in several places.
For example, the notion of being finite splits apart into several different constructive notions such as Bishop finiteness, subfiniteness, Kuratowski finiteness, etc.~\cite{Spiwack2010,Frumin2018}, and different notions seem to be needed to show that functions with finite support form an ordinal, to show that this ordinal satisfies the specification \eqref{exp-full-spec}, to show other expected properties of the exponential, and so on.

Classically, a function with finite support is equivalently given by the finite collection of input-output pairs where the function is greater than zero,
and this gives rise to a formulation that we found to be well behaved constructively.
The finite collection of input-output pairs can be represented as a list in which the input components are ordered decreasingly, which ensures that the representation is unique and that each input has at most one output.
Thus, our formulation corresponds to adopting Bishop finiteness, where we use the ordering to avoid repetitions.
In order to re-use results on ordinal multiplication, where the second component is dominant, i.e., $b_1 < b_2$ implies $(a_1,b_1) < (a_2,b_2)$, we swap the positions of inputs and outputs and consider lists of \emph{output-input} pairs.

\begin{definition}[\(\DL{\alpha}{\beta}\) \flink{Definition-19}]\label{DecrList}
  For ordinals $\alpha$ and $\beta$, we write
  \[
    \DL{\alpha}{\beta} \defeq
    \Subtype{l}{\List(\alpha \times \beta)}{\isdecreasing(\map {\pi_2} \, l)}
  \]
  for the type of lists over \(\alpha \times \beta\) decreasing in the \(\beta\)-component.
\end{definition}

\begin{remark}[\flink{Remark-20}]
  Since the type expressing that a list is decreasing in the second component is a
  proposition, it follows that two elements of \(\DL{\alpha}{\beta}\) are equal as
  soon as their underlying lists are equal.
  Accordingly, in denoting elements \((l, p)\) of type \(\DL{\alpha}{\beta}\), we will always omit the second proof component \(p\),
  and simply write \(l : \DL{\alpha}{\beta}\).
\end{remark}

Following Sierpi\'nski's construction, all outputs in the finite collection of
input-output pairs should be greater than the least element.
Therefore, we should be considering the type
\(
  \DL{\posalpha}{\beta}
\)
where, for 
$\alpha$ with least element $\bot$, we write
\[
	\posalpha \defeq \Subtype{a}{\alpha}{a > \bot}
\]
for the set of all elements greater than the least element.
In general, this subtype is not necessarily an ordinal:
\begin{proposition}[\flink{Proposition-21-i}]\label{positive-subset-ordinal-iff-lem}
  \LEM\ holds if and only if, for all ordinals~$\alpha$, the subtype of positive
  elements $\posalpha$ is an ordinal.
\end{proposition}
\begin{proof}
  It is not hard to check that \LEM\ allows one to prove that \(\posalpha\) is
  an ordinal.
  For the converse, we assume that \(\Ord_{>\Zero}\), the (large) subtype of
  ordinals strictly greater than \(\Zero\), is an ordinal.
  To prove \LEM\ it is enough to prove that the ordinal \(\Two\) of booleans and
  the ordinal \(\Omega\) of truth values (cf.~the example in \cref{sec:ord-in-hott}) are
  equal. So let us show that they have the same predecessors in \(\Ord_{> \Zero}\),
  namely only the one-element ordinal \(\One\).
  For \(\Two\) it is straightforward that its only predecessor in \(\Ord_{> \Zero}\) is \(\One\).
  For \(\Omega\), we note that if \(\Zero < \alpha < \Omega\), then
  \(\alpha = \Omega \initseg Q\) for some proposition \(Q\) and further we have
  \(\Zero < Q\), so that \(Q\) must hold, and hence
  \(\alpha = \Omega \initseg Q = \One\).
\end{proof}

To ensure that $\posalpha$ is an ordinal, and consequently $\DL{\posalpha}{\beta}$ as
well, it suffices to require the least element $\bot$ to be \emph{trichotomous},
meaning for all $x:\alpha$, either $x = \bot$ or $x > \bot$.
As pointed out to us by Paul Levy, a trichotomous least element is simply the
least element with respect to the ``disjunctive order'' $\leqslant$ defined by
$x \leqslant y \iff (x < y) + (x = y)$.

\begin{lemma}[\flink{Lemma-22-i}]\label{trichotomy-gives-positive-subset-ordinal}
    An ordinal $\alpha$ has a trichotomous least element if and only if $\alpha = \One + \alpha'$ for some (necessarily unique) ordinal~$\alpha'$. If this happens, then \(\alpha' = \posalpha\).
\end{lemma}
\begin{proof}
  Assume \(\alpha\) has a trichotomous least element \(\bot\).
  We first want to show that in this case \(\posalpha\) is an ordinal, and then
  that \(\alpha = \One + \posalpha\).
  By trichotomy of \(\bot\), we can prove that the order on $\posalpha$ inherited from~$\alpha$ is extensional, and thus that $\posalpha$ is an ordinal, since transitivity and wellfoundedness are always retained by the inherited order (\cref{subtype-ordinal}). Using trichotomy again, we can define an equivalence $\alpha \to \One + \posalpha$ by mapping $x:\alpha$ to the left if $x=\bot$ and to the right if $x > \bot$.

The converse is immediate, and uniqueness, i.e., \(\alpha' = \posalpha\),
follows as addition is left cancellable, as proven by Escard\'o~\cite[{\texttt{Ordinals.AdditionProperties}}]{TypeTopologyOrdinals}.
\end{proof}

If \(\alpha\) has a trichotomous least element, we thus have our candidate for a more concrete implementation of the exponential \(\alpha^\beta\); the following suggestion is similar to Grayson's~\cite{Grayson1982}, to which we come back in \cref{sec:grayson}.

\begin{definition}[Concrete exponentiation, \(\listexp{\alpha}{\beta}\) \flink{Definition-23}]
  For ordinals \(\alpha\) and \(\beta\) with \(\alpha\) having a trichotomous
  least element, we write $\listexp{\alpha}{\beta}$ for $\DL{\posalpha}{\beta}$
  (cf.~\cref{DecrList}) and call it the \emph{concrete exponentiation of $\alpha$
    and $\beta$}.
\end{definition}

Thanks to \cref{trichotomy-gives-positive-subset-ordinal}, we often choose
to work with the more convenient
\(\listexp{\One + \alpha'}{\beta} = \DL{\alpha'}{\beta}\) rather than
\(\listexp{\alpha}{\beta} = \DL{\posalpha}{\beta}\) in the Agda formalization.
We next prove that indeed \(\listexp{\alpha}{\beta}\) can be given a rather natural order which makes it into an ordinal.

\begin{proposition}[\flinkprime{1}{Proposition-24}]%
  \label{listexp-is-ordinal-with-trichotomous-least-element}
  For ordinals \(\alpha\) and \(\beta\) with \(\alpha\) having a trichotomous
  least element, the lexicographic order on lists makes $\listexp{\alpha}{\beta}$ into an
  ordinal that again has a trichotomous least element.
\end{proposition}
\begin{proof}
  By \cref{trichotomy-gives-positive-subset-ordinal}, $\posalpha$ is an ordinal,
  hence $\posalpha \times \beta$ is also an ordinal. The lexicographic order on
  $\mathsf{List}(\posalpha \times \beta)$ preserves transitivity of the
  underlying order. Using structural induction on lists, we can show that the
  lexicographic order on the subset of lists decreasing in the second component
  is extensional. For wellfoundedness, we describe the main argument below. In
  the end, $\listexp{\alpha}{\beta}$ is an ordinal and the empty list \(\nill\)
  is easily seen to be its trichotomous least element.

  For wellfoundedness, the following construction is key. Let \((X,<)\) be a
  type with a transitive and wellfounded relation and let \(\List_<(X)\) be the
  type of decreasing lists over \(X\) with the lexicographic order.
  We claim that if \(l\) is accessible in \(\List_<(X)\) and \(x \cons l\) is
  decreasing, then \(x \cons l\) is also accessible in \(\List_<(X)\).
  To prove the claim, we proceed by transfinite induction in \(X\). That is, we
  get to assume
  \begin{equation}\label{decr-list-wf-IH-1}\tag{$\textup{IH}_1$}
    \forall y : X, y < x \Rightarrow \pa*{\forall l : \List_<(X). \text{\(l\)
        is accessible \(\land\) \(y \cons l\) is decreasing \(\Rightarrow y \cons l\) is
        accessible}}.
  \end{equation}
  Now let \(l : \List_<(X)\) be accessible and \(x \cons l\) decreasing. We are
  to show that \(x \cons l\) is accessible.  By the accessibility assumption of
  \(l\), we may assume the following:
  \begin{equation}\label{decr-list-wf-IH-2}\tag{$\textup{IH}_2$}
    \forall l' : \List_<(X), l' < l \Rightarrow x \cons l' \text{ is accessible}.
  \end{equation}
  For the accessibility of \(x \cons l\), we need to prove that
  \(l'\) is accessible for every \(l' : \List_<(X)\) with
  \(l' < x \cons l\). We do this by list induction.
  So let \(l'\) be such a decreasing list.
  If \(l' \equiv \nill\), then \(l'\) is clearly accessible as it has no
  predecessors, so we are done.  If \(l' \equiv y \cons l''\), then, since
  \(l' < x \cons l\) by assumption, there are two cases to consider: either
  \(y < x\), or \(y = x\) and \(l'' < l\).
  In the second case, we are done by~\eqref{decr-list-wf-IH-2}.
  In the first case, since \(y < x\) and \(y \cons l''\) is decreasing,
  we have \(l'' < x \cons l\). Thus, \(l''\) is accessible by the (list)
  induction hypothesis, and we are done by~\eqref{decr-list-wf-IH-1}.
\end{proof}

We now wish to characterize the initial segments of the ordinal
\(\listexp{\alpha}{\beta}\). Before doing so we must introduce a lemma and two instrumental
functions (\cref{exp-segment-inclusion,exp-tail}).

\begin{lemma}[\flinkprime{1}{Lemma-25-i}]\label{listexp-monotone}
  Let $\alpha$ be an ordinal with a trichotomous least element. Any order preserving map \(f : \beta \to \gamma\) induces an order preserving
  map \(\overline{f} : \listexp{\alpha}{\beta} \to \listexp{\alpha}{\gamma}\) by
  applying \(f\) to the second component of each pair in the list.

  Moreover, if \(f\) is a simulation, then so is \(\overline f\).
  Consequently, \(\listexp{\alpha}{-}\) is monotone.
\end{lemma}
\begin{proof}
  Note that order preservation of \(f\) ensures that the outputs of
  \(\overline f\) are again decreasing in the second component, so that we have
  a well defined map which is easily seen to be order preserving.
  Now suppose that \(f\) is moreover a simulation. Since \(f\) is order
  reflecting and injective (\cref{simulation-facts}), it follows that
  \(\overline f\) is order reflecting. Therefore, it suffices to prove that if
  we have \(l < \overline{f}\,l_1\), then there is
  \(l_2 : \listexp{\alpha}{\beta}\) with \(\overline f\,l_2 = l\).
  We do so by list induction on \(l\). The case \(l = \nill\) is easy and if \(l\) is
  a singleton, then we need only use that \(f\) is a simulation.
  So let \(l = (a,c) \cons (a',c') \cons l'\) and
  \(l_1 = (a_1,b_1) \cons l_1'\).
  We proceed by case analysis on \(l < \overline f\,l_1\) and work out the
  details in case \(c < f\,b_1\); the other cases are dealt with similarly.
  Since \(f\) is a simulation we have \(b_2 : \beta\) such that
  \(f \, b_2 = c\).
  Since \(c' < c < f\,b_1\) we have \(\left((a',c') \cons l'\right) < \overline f\,l_1\), and
  hence we get \(l_2'\) such that \(\overline f\, l_2' = (a',c') \cons l'\) by
  induction hypothesis.
  Since \(c' < c = f\,b_2\) and \(f\) is order reflecting, the list
  \(l_2 \defeq (a,b_2) \cons l_2'\) is decreasing in the second component and by
  construction we have \(\overline f\,l_2 = l\) as desired.
\end{proof}

In particular, the construction in \cref{listexp-monotone} gives us
\begin{fequationprime}{2}{Eq-4}\label{exp-segment-inclusion}
  \iota_b : \listexp{\alpha}{\beta \initseg b} \le \listexp{\alpha}{\beta}.
\end{fequationprime}
for every \(b : \beta\).
In a sense, this map has an inverse: given a list
\(l : \listexp{\alpha}{\beta}\) such that each second component is below some
element \(b : \beta\), we can construct
\begin{fequationprime}{2}{Eq-5}\label{exp-tail}
  \tau_{b}\,l : \listexp{\alpha}{\beta \initseg b}
\end{fequationprime}
by inserting the required inequality proofs.
Moreover, the assignment \(l \mapsto \tau_{b}\,l\) is order preserving and inverse to \(\iota_{b}\).

\begin{proposition}%
  [Initial segments of \(\listexp{\alpha}{\beta}\), \flinkprime{2}{Proposition-26-i}]%
  \label{initial-segment-of-listexp}
  For ordinals \(\alpha\) and \(\beta\) with \(\alpha\) having a trichotomous
  least element, we have
  \begin{equation*}
      {\listexp{\alpha}{\beta} \initseg \left((a,b) \cons l\right)}
      = {\listexp{\alpha}{\beta \initseg b} \times (\alpha \initseg a) +
      \listexp{\alpha}{\beta \initseg b} \initseg \tau_b\,l}.
  \end{equation*}

  Similarly, for \(a : \posalpha\), \(b : \beta\) and
  \(l : \listexp{\alpha}{\beta\initseg b}\), we have
  \begin{equation*}
      {\listexp{\alpha}{\beta} \initseg \left((a,b) \cons {\iota_b\,l}\right)}
      = {\listexp{\alpha}{\beta \initseg b} \times (\alpha \initseg a) +
      \listexp{\alpha}{\beta \initseg b} \initseg l}.
  \end{equation*}
\end{proposition}

\begin{proof}
  The second claim follows from the first and the fact that \(\tau_b\) and
  \(\iota_b\) are inverses.

  For the first claim, we construct order preserving functions \(f\) and
  \(g\) in each direction, and show that they are inverse to each other.
  We define an order preserving map
  \[
    f : {\listexp{\alpha}{\beta} \initseg \left((a,b) \cons l\right)}
    \to {\listexp{\alpha}{\beta \initseg b} \times (\alpha \initseg a) +
    \listexp{\alpha}{\beta \initseg b} \initseg {\tau_b\,l}}
  \]
  by cases on \(l_0 < \left((a,b) \cons l\right)\) via
  \[
  f\,l_0 \defeq \begin{cases}
     \inl (\nill,\bot)
     &\text{if \(l_0 = \nill\)}; \\
     \inl ((a',b') \cons \tau_b\,l_1 , \bot)
     &\text{if \(l_0 = (a',b') \cons l_1\) and \(b' < b\)}; \\
     \inl (\tau_b\,l_1, a')
     &\text{if \(l_0 = (a',b) \cons l_1\) and \(a' < a\)}; \\
     \inr (\tau_b\,l_1)
     &\text{if \(l_0 = (a,b) \cons l_1\) and \(l_1 < l\)}.
  \end{cases}
  \]

  In the other direction, we define \(g\), using the fact that equality with
  \(\bot\) is decidable, by
  \begin{align*}
    g\,(\inl (l_1,a')) &\defeq
    \begin{cases}
      \iota_b \, l_1 &\text{if \(a' = \bot\)}; \\
      (a',b) \cons \iota_b \, l_1 &\text{if \(a' > \bot\)};
    \end{cases} \\
          g\,(\inr l_1) &\defeq (a,b) \cons \iota_b \,l_1.
  \end{align*}
  Direct calculations then verify that \(g\) is order preserving and that
  \(f \circ g = \id\) and \(g \circ f = \id\).
\end{proof}

The upcoming \cref{concrete-exp-satisfies-spec,listexp-by-+} can be
derived from the corresponding facts about abstract exponentiation and
\cref{abstrexp-listexp-coincide} below, but for comparison we include sketches
of direct proofs as well (for further details, we refer the interested reader to
the formalization).
Alternatively, \cref{abstrexp-listexp-coincide} can be derived from \cref{concrete-exp-satisfies-spec} and the fact that operations satisfying the specification \eqref{exp-strong-spec} are unique --- the proof effort is about the same for both strategies.

\begin{theorem}[\flinkprime{3}{Theorem-27}]\label{concrete-exp-satisfies-spec}
  Concrete exponentiation \(\listexp{\alpha}{\beta}\) satisfies the specification \eqref{exp-strong-spec} \textup{(}and hence the specification \eqref{exp-full-spec}\textup{)} for $\alpha$ with a trichotomous least element.
\end{theorem}
\begin{proof}[Proof sketch]
  Again, \eqref{exp-full-spec} is equivalent to \eqref{exp-strong-spec} since $\alpha$ is assumed to be positive.
  The successor case follows from the more general \cref{listexp-by-+} below.
  For the supremum clause of \eqref{exp-strong-spec}, i.e., to prove
  \(\listexp{\alpha}{\sup F_{\sbullet}} = \One \vee \sup\listexp{\alpha}{F_{\sbullet}}\), we appeal to \cref{listexp-monotone} to obtain
  simulations
  \[
    \sigma_i : \listexp{\alpha}{F_i} \le \listexp{\alpha}{\sup F_{\sbullet}},
  \]
  yielding a simulation
  \(\sigma : \One \vee \sup\listexp{\alpha}{F_{\sbullet}} \le \listexp{\alpha}{\sup F_{\sbullet}}\).
  We now show that \(\sigma\) is additionally a surjection, and hence an ordinal
  equivalence.
  The key observation is the following claim, which follows
  by induction on lists and \cref{initial-segment-of-sup}:
  \begin{claim}
    For every list \((a , [i , x]) \cons l\) of type
    \(\listexp{\alpha}{\sup F_{\sbullet}}\), there exists a list
    \(l' : \listexp{\alpha}{F_i}\) such that
    \(\sigma_i \pa*{(a , x) \cons l'} = (a,[i,x]) \cons l\).
  \end{claim}
  Now, for the surjectivity of \(\sigma\), let
  \(\ell : \listexp{\alpha}{\sup F_{\sbullet}}\) be arbitrary.
  If \(\ell = \nill\), it is the image of the unique element $\ast : \One$.
  If \(\ell = (a,y) \cons l\), then by \cref{initial-segment-of-sup},
  \(y = [i,x]\) for some \(i\) and \(x\), and we apply the claim to get \(l'\)
  with \(\sigma[i,(a,x) \cons l'] = \sigma_i((a,x)\cons l') = \ell\).
\end{proof}

\begin{proposition}[\flinkprime{3}{Proposition-28}]\label{listexp-by-+}
  For ordinals \(\alpha\), \(\beta\) and \(\gamma\) with \(\alpha\) having a
  trichotomous least element, we have
  $\listexp{\alpha}{\beta + \gamma} = \listexp{\alpha}{\beta} \times
  \listexp{\alpha}{\gamma}$.
\end{proposition}
\begin{proof}
  We construct order preserving maps in both directions, and prove that they are inverse to each other. The harder direction is from left to right. It is straightforward to define a function on underlying lists \(f : \listexp{\alpha}{\beta + \gamma} \to \listexp{\alpha}{\beta} \times \listexp{\alpha}{\gamma}\) as follows:
  \begin{align*}
    f : \listexp{\alpha}{\beta + \gamma}
        &\to \listexp{\alpha}{\beta} \times \listexp{\alpha}{\gamma} \\
    \nill &\mapsto (\nill , \nill); \\
    (a , \inl b) \cons l &\mapsto \pa*{(a , b) \cons \pi_1(f\, l),
    \pi_2(f\,l)}; \\
    (a , \inr c) \cons l &\mapsto \pa*{\pi_1(f\, l),
    (a , c) \cons \pi_2(f\,l)}.
  \end{align*}
  The fact that this is well defined, in particular that \(\pi_2(f\,l)\) yields
  a list that is decreasing in the second component, follows from the
  observation that an element of \(\DL{\alpha}{\beta+\gamma}\) starting with
  \((a , \inl b)\) cannot have any \((a' , \inr c)\) entries. Put differently,
  \(\pi_2(f ((a , \inl b) \cons l)) = \nill\) for any \(l\).
  A~proof by case analysis shows that \(f\) is order preserving.

  In the other direction, we define
  \begin{align*}
    g : \listexp{\alpha}{\beta} \times \listexp{\alpha}{\gamma}
      &\to \listexp{\alpha}{\beta + \gamma} \\
    (\nill , \nill) &\mapsto \nill ;\\
    ((a , b) \cons l_1 , \nill) &\mapsto (a,\inl b)\cons g(l_1,\nill);\\
   (l_1,(a,c)\cons l_2) &\mapsto (a , \inr c) \cons g (l_1,l_2).
  \end{align*}
  We again prove by a case analysis that \(g\) is order preserving, and that \(g\)  is the inverse of \(f\).
\end{proof}

A feature of concrete exponentiation is that it preserves decidability properties.
\begin{proposition}[\flinkprime{3}{Proposition-29}]\label{listexp-discrete}
  Assume $\alpha$ has a trichotomous least element. If $\alpha$ and $\beta$ have
  decidable equality, then so does $\listexp{\alpha}{\beta}$.
\end{proposition}
\begin{proof}
  All of $\times$, \(\List\), and taking subtypes preserve decidable equality, and $\listexp{\alpha}{\beta}$ is a subtype of \(\List\,(\alpha \times \beta)\).
\end{proof}

We recall that an ordinal \(\alpha\) is said to be \emph{trichotomous} if we
have \((x < y) + (x = y) + (y < x)\) for every \(x,y : \alpha\).
\begin{proposition}[\flink{Proposition-30}]\label{listexp-trichotomous}
  For an ordinal $\alpha$ with a least element, if $\alpha$ and $\beta$ are
  trichotomous, then so is $\listexp{\alpha}{\beta}$.
\end{proposition}
\begin{proof}
  Proved by a straightforward induction on lists.
\end{proof}

\section{Abstract and Concrete Exponentiation}\label{sec:equivalence}

Since both the abstract and concrete constructions of ordinal exponentiation
satisfy the specification~\eqref{exp-strong-spec}, cf.~\cref{rem:f-unique}, they
in fact coincide whenever the base has a trichotomous least element.
We give an alternative proof based on initial segments
(\cref{abstrexp-listexp-coincide}), and explain its computational content
by showing how it relates to a surjective denotation function, which represents
elements of the abstract exponential as lists of the concrete exponential.

\subsection{Abstract and Concrete Exponentiation Coincide}

If \(\alpha\) has a trichotomous least element, then \(\alpha\) in particular has a least element, i.e., \(\One \leq \alpha\). Hence both the abstract exponentiation \(\alpha^\beta\) and the concrete exponentiation \(\listexp{\alpha}{\beta}\) are well defined and well behaved in this case.

\begin{theorem}[\flink{Theorem-31}]\label{abstrexp-listexp-coincide}
  For all ordinals \(\alpha\) and \(\beta\) such that \(\alpha\) has a
  trichotomous least element, we have
  \[
    {\abstrexp{\alpha}{\beta}} = {\listexp{\alpha}{\beta}}.
  \]
\end{theorem}
\begin{proof}
  Let \(\bot\) be the trichotomous least element of \(\alpha\). We prove the equation by transfinite induction on \(\beta\). Our induction hypothesis reads:
  \begin{equation}\tag{IH}\label{coincide-IH}
    \forall (b : \beta) .\, \abstrexp{\alpha}{\beta \initseg b} =
    \listexp{\alpha}{\beta \initseg b}.
  \end{equation}
  These equalities induce simulations and simulations preserve initial segments
  (\cref{simulation-facts}), so for all \(b : \beta\), the
  simulation provides for every \(e : \abstrexp{\alpha}{\beta\initseg b}\) a
  unique \(l : \listexp{\alpha}{\beta\initseg b}\) with
  \(\abstrexp{\alpha}{\beta\initseg b} \initseg e = \listexp{\alpha}{\beta
    \initseg b} \initseg l\), and similarly if we start with an element of
  \(\listexp{\alpha}{\beta \initseg b}\) instead.

  By extensionality of the ordinal of ordinals, it suffices to show that
  each initial segment of \(\abstrexp{\alpha}{\beta}\) is equal to an
  initial segment of \(\listexp{\alpha}{\beta}\) and vice versa.

  Suppose first that we have \(e_0 : \abstrexp{\alpha}{\beta}\). By
  \cref{initial-segment-of-sup} and the fact that we are proving a proposition,
  we have \(e_0 = [\inl\star,\star] \equiv \bot\) or \(e_0 = [\inr b,(e,a)]\)
  with \(a : \alpha\) and \(e : \abstrexp{\alpha}{\beta \initseg b}\).
  In the first case, we have
  \({\abstrexp{\alpha}{\beta} \initseg e_0} = \Zero = \listexp{\alpha}{\beta} \initseg
  \nill\). The second case has two subcases:
  \(a\) is either equal to the trichotomous least
  element \(\bot\), or greater than it.
  If \(a = \bot\), then we calculate
  \begin{align*}
    &\abstrexp{\alpha}{\beta} \initseg e_0 \\
    &= \abstrexp{\alpha}{\beta \initseg b} \initseg e
    &&\text{(by \cref{initial-segment-of-abstrexp})} \\
    &= \listexp{\alpha}{\beta \initseg b} \initseg l
    &&\text{(for a unique \(l\) by \ref{coincide-IH})} \\
    &= \listexp{\alpha}{\beta} \initseg {\iota_b\,l}
    &&\text{(using \cref{simulation-facts}
       and \(\iota_b\) from \cref{exp-segment-inclusion})}
  \end{align*}
  completing the proof for this case.
  If \(a > \bot\), then we calculate
  \begin{align*}
    &\abstrexp{\alpha}{\beta} \initseg e_0 \\
    &= \abstrexp{\alpha}{\beta \initseg b}
      \times (\alpha \initseg a)
      + \abstrexp{\alpha}{\beta \initseg b} \initseg e
    &&\text{(by \cref{initial-segment-of-abstrexp})} \\
    &= \listexp{\alpha}{\beta \initseg b}
      \times (\alpha \initseg a)
    + \listexp{\alpha}{\beta \initseg b} \initseg l
    &&\text{(for a unique \(l\) by \ref{coincide-IH})} \\
    &= \listexp{\alpha}{\beta} \initseg \left((a, b) \cons \iota_b\,l \right)
    &&\text{(by \cref{initial-segment-of-listexp})}
  \end{align*}
  completing the proof for this case.

  Now let \(l_0 : \listexp{\alpha}{\beta}\). Then either \(l_0 = \nill\) in
  which case we are done, because
  \(\listexp{\alpha}{\beta} \initseg l_0 = \Zero = \abstrexp{\alpha}{\beta}
  \initseg \bot\), or \(l_0 = (a,b) \cons l\). In this second case, we calculate
  \begin{align*}
    &\listexp{\alpha}{\beta} \initseg {\left((a,b) \cons l\right)} \\
    &= \listexp{\alpha}{\beta \initseg b} \times (\alpha \initseg a)
	+
      \listexp{\alpha}{\beta \initseg b} \initseg {\tau_b\,l}
    &\text{(by \cref{initial-segment-of-listexp})} \\
    &= {\abstrexp{\alpha}{\beta \initseg b} \times (\alpha \initseg a) +
      \abstrexp{\alpha}{\beta \initseg b} \initseg e}
    &\text{(for a unique \(e\) by \ref{coincide-IH})} \\
    &= \abstrexp{\alpha}{\beta} \initseg {[\inr b,(e,a)]}
    &\text{(by \cref{initial-segment-of-abstrexp})}
  \end{align*}
  finishing the proof.
\end{proof}

The following decidability properties follow directly from
\cref{abstrexp-listexp-coincide} and
\cref{listexp-discrete,listexp-trichotomous}.
\begin{corollary}[\flink{Corollary-32-i}]%
  \label{abstrexp-decidability-properties}\leavevmode
  \begin{enumerate}[label=(\roman*)]
  \item Suppose that \(\alpha\) has a trichotomous least element.  If \(\alpha\)
    and \(\beta\) have decidable equality, then so does
    \(\abstrexp{\alpha}{\beta}\).
  \item Suppose \(\alpha\) has a least element.
    If \(\alpha\) and \(\beta\) are trichotomous, then so is
    \(\abstrexp{\alpha}{\beta}\). \qedNoProof
  \end{enumerate}
\end{corollary}

Before we knew that \cref{abstrexp-listexp-coincide} was true, we started working on a direct proof that repeated concrete exponentiation is exponentiation by the product, i.e., \(\listexp{\listexp{\alpha}{\beta}}{\gamma} = \listexp{\alpha}{\beta \times \gamma}\). However, dealing with all of the side conditions stating that lists are decreasing proved to be too tedious for us to finish the construction. Fortunately, this result follows for free via \cref{abstrexp-by-product,abstrexp-listexp-coincide}.

\begin{corollary}[\flinkprime{4}{Corollary-33}]\label{listexp-by-product}
  For ordinals \(\alpha\), \(\beta\) and \(\gamma\) with \(\alpha\) having a
  trichotomous least element, we have
  \(
  \listexp{\alpha}{\beta \times \gamma} =
  \listexp{\listexp{\alpha}{\beta}}{\gamma}
  \). \qedNoProof
\end{corollary}

Note that in the above corollary we implicitly used that \(\exp(\alpha,\beta)\)
has a trichotomous least
element~(\cref{listexp-is-ordinal-with-trichotomous-least-element}).

\subsection{Lists as Representations}\label{subsec:lists-as-representations}

As remarked above, we used the assumption that the least element of $\alpha$ is trichotomous to show that the set of lists $\DL{\alpha}{\beta}$ is an ordinal.
But even without this assumption, we can view such a list as a \emph{representation} of something in the abstract exponential~$\abstrexp{\alpha}{\beta}$.
This is made precise by the following denotation function.

\begin{definition}[Denotation function, \(\denotes{-}{\beta}\) \flinkprime{5}{Definition-34}]
  We define
  \[
   \denotes - \beta : \DL{\alpha}{\beta} \to \abstrexp{\alpha}{\beta}
  \]
  for any $\alpha$ by transfinite induction on $\beta$:
  \[
    \begin{array}{ccc}
      \denotes {\;\nill\;} \beta  & \defeq & \bot;\\
      \denotes {(a,b) \cons l} \beta & \defeq & [\inr b, (\denotes {\tau_b\,l} {\beta \initseg b}  , a)];
    \end{array}
  \]
  with \(\tau_b\) as in \cref{exp-tail}.
\end{definition}

\begin{remark}
  Note that the above definition is not by recursion on the
  list, as the recursive call in the non-empty list case is on \(\tau_b\,l\), which is not directly structurally smaller than \((a,b) \cons l\). With more work, \(\denotes{l}{\beta}\) could be defined by induction on the \emph{length} of the list \(l\), but a construction by transfinite induction on \(\beta\) with a non-recursive case-split on \(l\) is more straightforward.
\end{remark}

Every element of the abstract exponential $\abstrexp{\alpha}{\beta}$ merely has
a representation as a list, in the following sense:
\begin{proposition}[\flinkprime{5}{Proposition-36}]\label{denotes-surjective}
  For all ordinals \(\alpha\) and \(\beta\), the denotation function
  $\denotes - \beta$ is surjective.
\end{proposition}
\begin{proof}
	We do transfinite induction on $\beta$.
	For every $x : \abstrexp{\alpha}{\beta}$ we need to show that there exists a list $l$ with $\denotes{l}{\beta} = x$.
	As the goal is a proposition, we can do case distinction on $x$; for $x = [\inl \star,\star]$, the list is $\nill$, which leaves us with the case that $x$ is given by $b : \beta$ together with $e : \abstrexp{\alpha}{\beta \initseg b}$ and $a : \alpha$.
	By the induction hypothesis, there exists a list $l' : \DL{\alpha}{\beta\initseg b}$ whose denotation is $e$, ensuring that $(a,b) \cons \iota_b\,l'$ represents $x$.
\end{proof}

Note that \cref{denotes-surjective} does not assume that \(\alpha\) has a least element, and definitely not that it has a trichotomous least element.
However, when $\alpha$ does have a trichotomous least element, the equality established in
\cref{abstrexp-listexp-coincide}
induces a map from concrete to abstract exponentials,
\[
\contoabs : \listexp{\alpha}{\beta} \to \abstrexp{\alpha}{\beta},
\]
and we could hope to relate this function to the denotation function \(\denotes{-}{\beta} : \DL{\alpha}{\beta} \to \abstrexp{\alpha}{\beta}\).

Thanks to the trichotomous least element, we can normalize a list by removing those pairs which have the least element in the $\alpha$-component, yielding a map
\[
\normalize : \DL{\alpha}{\beta} \to \DL{\posalpha}{\beta}
\]
with
\[
  \normalize\,{((a,b)\cons l)} \defeq
  \begin{cases}
    \normalize\,l\ & \text{if \(a = \bot\)}; \\
    (a, b) \cons \normalize\,l & \text{if \(a > \bot\)}.
  \end{cases}
\]
Note that the codomain of \(\normalize\) is exactly \(\listexp{\alpha}{\beta}\).

The normalization function allows us to compare the induced map with the denotation function:
\begin{theorem}[\flinkprime{6}{Theorem-37}]\label{related-denotations}
  In case $\alpha$ has a trichotomous least element, the denotation function
  coincides with the equality between abstract and concrete exponentiation in the
  following sense:
  \[
    \denotes{-}{\beta} = \contoabs \circ \normalize.
  \]
\end{theorem}

The two following lemmas directly prove the theorem.

\begin{lemma}[\flinkprime{6}{Lemma-38}]\label{contoabs-is-denotation-function}
  The induced map \(\contoabs\) coincides with a denotation function
  \[
    \denotesprime - \beta : \DL{\posalpha}{\beta} \to \abstrexp{\alpha}{\beta}
  \]
  that is defined by transfinite induction just like \(\denotes - \beta\) but with a
  restricted domain instead.
\end{lemma}
\begin{proof}
  We prove \(\contoabs\,l = \denotesprime{l}{\beta}\) for all \(l : \DL{\posalpha}{\beta}\) by induction on \(\beta\) and a case distinction on \(l\), and use that elements of ordinals are equal if
  and only if they determine the same initial segments.
  The case of the empty list is easy, as it is the least element of
  \(\DL{\posalpha}{\beta}\), and since \(\contoabs\) is a simulation it must map
  it to the least element of \(\abstrexp{\alpha}{\beta}\).
  For nonempty lists, we have
  \begin{align*}
    &\abstrexp{\alpha}{\beta} \initseg {\contoabs\,((a,b) \cons l)} \\
    &= \listexp{\alpha}{\beta} \initseg {\left((a,b) \cons l\right)}
    &&\text{(since simulations preserve initial segments)} \\
    &= \listexp{\alpha}{\beta \initseg b} \times (\alpha \initseg a)
    + \listexp{\alpha}{\beta \initseg b} \initseg {\tau_b\,l}
    &&\text{(by \cref{initial-segment-of-listexp})} \\
    &= \abstrexp{\alpha}{\beta \initseg b} \times (\alpha \initseg a)
    + \listexp{\alpha}{\beta \initseg b} \initseg {\tau_b\,l}
    &&\text{(by \cref{abstrexp-listexp-coincide})} \\
    &=  \abstrexp{\alpha}{\beta \initseg b} \times (\alpha \initseg a)
    + \abstrexp{\alpha}{\beta \initseg b} \initseg {\contoabs\,(\tau_b\,l)}
    &&\text{(since simulations preserve initial segments)} \\
    &=  \abstrexp{\alpha}{\beta \initseg b} \times (\alpha \initseg a)
    + \abstrexp{\alpha}{\beta \initseg b} \initseg {\denotesprime {\tau_b\,l} {\beta \initseg b}}
    &&\text{(by IH)} \\
    &= \abstrexp{\alpha}{\beta} \initseg [\inr b, (\denotesprime{\tau_b\,l}{\beta\initseg b},a)]
    &&\text{(by \cref{initial-segment-of-abstrexp})} \\
    &\equiv \abstrexp{\alpha}{\beta} \initseg \denotesprime{(a,b) \cons l}{\beta}.&&\hspace{-1em}\qedhere
  \end{align*}
\end{proof}

\begin{lemma}[\flinkprime{6}{Lemma-39}]
  The denotations are related via normalization:
  \[
    \denotes{-}{\beta} = \denotesprime{-}{\beta} \circ \normalize.
  \]
\end{lemma}
\begin{proof}
  This is proved by transfinite induction on \(\beta\). The case of the empty
  list is easy.
  So consider an element of \(\DL{\alpha}{\beta}\) of the form \((a,b) \cons l\).
  Suppose first that \(a\) is not the trichotomous least element~\(\bot\), then
  \begin{align*}
    \denotes {(a,b) \cons l} {\beta}
    &\equiv [\inr b , (\denotes {\tau_b\,l} {\beta\initseg b},a)] \\
    &= [\inr b , (\denotesprime {\normalize(\tau_b\,l)} {\beta\initseg b},a)] &\text{(by IH)} \\
    &= [\inr b , (\denotesprime {\tau_b\,(\normalize\,l)} {\beta\initseg b},a)] \\
    &\equiv \denotesprime {\normalize {((a,b)\cons l)}} {\beta},
  \end{align*}
  where the penultimate equality uses that \(\normalize\) and \(\tau_b\) commute.

  For the trichotomous least element \(\bot\), we have
  \begin{align*}
    \denotes {(\bot,b) \cons l} {\beta}
    &\equiv [\inr b , (\denotes {\tau_b\,l} {\beta\initseg b},\bot)] \\
    &= \denotesprime {\iota_b(\normalize (\tau_b\,l))} {\beta}
      &&(\ast) \\
    &= \denotesprime {\normalize (\iota_b(\tau_b\,l))} {\beta}
      &&\text{(since \(\normalize\) and \(\iota_b\) commute)} \\
    &= \denotesprime {\normalize\,l} {\beta}
      &&\text{(since \(\tau_b\) is a section of \(\iota_b\))} \\
    &\equiv \denotesprime {\normalize((\bot,b) \cons l)} {\beta},
  \end{align*}
  where \((\ast)\) holds because these elements determine the same initial segments:
  \begin{align*}
    \abstrexp{\alpha}{\beta} \initseg
    [\inr b , (\denotes {\tau_b\,l} {\beta\initseg b},\bot)]
    &= \abstrexp{\alpha}{\beta \initseg b} \times (\alpha \initseg \bot) +
      \abstrexp{\alpha}{\beta \initseg b} \initseg
      {\denotes {\tau_b\,l} {\beta\initseg b}}
    &&\text{(by \cref{initial-segment-of-abstrexp})} \\
    &= \abstrexp{\alpha}{\beta \initseg b} \initseg
      {\denotes {\tau_b\,l} {\beta\initseg b}}
    &&\text{(as \(\alpha \initseg \bot = \Zero\))} \\
    &= \abstrexp{\alpha}{\beta \initseg b} \initseg
      {\denotesprime {\normalize(\tau_b\,l)} {\beta\initseg b}}
    &&\text{(by IH)} \\
    &= \listexp{\alpha}{\beta \initseg b} \initseg {\normalize(\tau_b\,l)}
    &&\text{(by \ref{simulations-preserve-initial-segments} in \cref{simulation-facts})} \\
    &= \abstrexp{\alpha}{\beta} \initseg
      {\denotesprime {\iota_b(\normalize(\tau_b\,l))} {\beta}},
    &&\text{(by \ref{simulations-preserve-initial-segments} in \cref{simulation-facts})}
  \end{align*}
  where the second-to-last equality employs the map
  \(\denotesprime{-}{\beta \initseg b}\), which is a simulation by
  \cref{contoabs-is-denotation-function} as \(\contoabs\) is induced by an
  equality (and hence is a simulation). For the final equality we consider
  the composition of simulations \(\denotesprime{-}{\beta} \circ \iota_b\).
\end{proof}

\begin{remark}
  We do not know of a \emph{direct} proof that the function \(\contoabs\) is an
  ordinal equivalence, or even a simulation.
  Instead, our proof of \cref{abstrexp-listexp-coincide} makes use of the
  equalities given by the induction hypothesis to prove the inductive step.
\end{remark}

\subsection{On Grayson's Decreasing Lists}\label{sec:grayson}

A slight variation of our decreasing list construction \(\listexp{\alpha}{\beta}\) was suggested in Grayson's PhD thesis~\cite[\S\@IX.3]{Grayson1978}, the relevant part of which has appeared as Grayson~\cite[\S 3.2]{Grayson1982}.
His setting is an unspecified version of constructive set (or type) theory and uses setoids (i.e., sets with an equivalence relation) for the construction.
Grayson's suggestion does not require the base ordinal~\(\alpha\) to have a least element.
Unfortunately, his construction (which is presented without proofs) does not work in the claimed generality, as assuming that it always yields an ordinal is equivalent to the law of excluded middle.
We present our argument for this
in the setting of the current paper, but the argument carries over to a foundation based on setoids.

\begin{definition}[Grayson~{\cite[\S\@IX.3]{Grayson1978}} \flink{Definition-41}]
Given an ordinal \(\alpha\), we say that \(x:\alpha\) is \emph{positively non-minimal} if $\exists(a : \alpha). x > a$.
Then, given a second ordinal \(\beta\), the type of \emph{Grayson lists}, written \(\grayson \alpha \beta\), is the type of lists over \(\alpha \times \beta\) that are decreasing in the \(\beta\)-component and where all \(\alpha\)-components are positively non-minimal.
\end{definition}

If \(\alpha\) has a trichotomous least element \(\bot\), it is easy to see (and formalized~\cite[{\texttt{Grayson}}]{formalization}) that \(x\) is positively non-minimal if and only if \(x > \bot\), i.e.\ the two notions of positivity coincide.
Thus, under this assumption, \(\grayson \alpha \beta\) becomes equivalent to \(\listexp{\alpha}{\beta}\).
However, the assumption cannot be removed:

\begin{proposition}[\flink{Proposition-42-i}]
	\(\LEM\) holds if and only if \(\grayson \alpha \beta\) is an ordinal for all (possibly large) ordinals \(\alpha\) and \(\beta\).
	This remains true 
        with the additional condition that \(\alpha\) has a least element.
\end{proposition}
\begin{proof}
	If \(\LEM\) is assumed, then every \(\alpha\) is empty or has a trichotomous least element, and the construction works in either case. For the other direction, let us fix \(\beta\) to be \(\One\); then, \(\grayson \alpha \One\) is equivalent to \(\One + \alpha^+\), where the latter is the type of positively non-minimal elements of \(\alpha\).
	It is easy to see that \(\grayson \alpha \One\) is an ordinal if and only if \(\alpha^+\) is.
	But the assumption that \(\Ord^+\) is an ordinal implies \(\LEM\) (cf.~the proof of~
        \cref{positive-subset-ordinal-iff-lem}).
\end{proof}

Finally, Grayson~\cite[\S\@IX.3]{Grayson1978} claims the recursive equation
\begin{equation*}
	\grayson \alpha \beta = \One \vee \sup_{b : \beta}(\grayson{\alpha}{\beta \downarrow b} \times \alpha)
\end{equation*}
in full generality. It follows from the above discussion and \cref{abstrexp-listexp-coincide} that this equation holds if \(\alpha\) has a trichotomous least element. The latter condition is indispensable, as the left-hand side always has a trichotomous least element (the empty list), while, for \(\beta = \One\) and \(\alpha \geq \One\), the right-hand side does so exactly if \(\alpha\) does.
For more details, we refer the interested reader to the \texttt{Grayson} file of
our formalization~\cite{formalization}.

\section{Abstract Cancellation Arithmetic}\label{sec:cancellation}

Results on \emph{cancellation} properties are useful tools; they tell us that a function can be ``undone'' or ``reversed'' in certain situations.
The more precise question is whether a given function reflects a certain property that we are interested in.
Here, we discuss whether the arithmetic operations $(\alpha + \_)$, $(\alpha \times \_)$ and $(\alpha^-)$, for fixed $\alpha$, reflect the order $\leq$; if they do, they automatically reflect equality.

Concretely, for addition, there are two questions:
\begin{enumerate}
	\item Does $\beta \leq \gamma$ imply $\alpha + \beta \leq \alpha + \gamma$?
\end{enumerate}
This holds; it is monotonicity of addition on the right.
\begin{enumerate}[resume]
	\item Does $\alpha + \beta \leq \alpha + \gamma$ imply $\beta \leq \gamma$?
\end{enumerate}
This also holds and is not so hard to prove. However, the analogous questions for multiplication and exponentiation are non-trivial.
In this section, we study the question for a general abstract framework that has addition, multiplication, and exponentiation as instances.
While this section establishes results on left cancellation, it is worth noting that the symmetric question about right cancellation, i.e.\ cancellation of $(\_ + \alpha)$, $(\_ \times \alpha)$ and $(\_^\alpha)$, can only have a negative answer.
An easy way to see this is to observe that \emph{all} the following expressions are equal to $\omega$:
\begin{fequation}{Eq-6}
  \begin{array}{ccc}
    \Zero + \omega = \omega \qquad& \One \times \omega = \omega \qquad& \Two^\omega = \omega \\
    \One + \omega = \omega \qquad& \Two \times \omega = \omega \qquad& \Three^\omega = \omega
  \end{array}
\end{fequation}
Right cancellation for addition, multiplication or exponentiation would thus imply that \(\Zero = \One\), \(\One = \Two\) and \(\Two = \Three\) respectively, which would be absurd.

\subsection{Order-Preserving Functions, and Classical Comparisons}

Let us recall the following result:
\begin{lemma}[{Escard\'o~\cite[{\texttt{Ordinals.OrdinalOfOrdinals}}]{TypeTopologyOrdinals} \flink{Lemma-43}}]\label{lem:escardo-op-map}
	If $\beta < \alpha$, then there is no order preserving map $f : \alpha \to \beta$. \qed
\end{lemma}
In the current subsection (which the reader may wish to skip during the first read), we extend this result for later usage.

If we assume $\LEM$, then any order-preserving function $f$ gives rise to a simulation; thus, classically, an order-preserving function $f$ is no less useful than a simulation $\alpha \leq \beta$.
This motivates writing $\alpha \incr \beta$ for the type of order-preserving functions.
Similarly as to how one can weaken $\leq$~to~$\incr$, we can weaken $<$ on $\Ord$ to the classically equivalent notion $\incrst$,
where an element of $\alpha \incrst \beta$ is a quadruple of a function $f : \alpha \to \beta$, a proof that $f$ is order-preserving, an element $b_0 : \beta$, and a proof that $f$ ``stays under $b_0$'', i.e., $f$ factors through $\beta \downarrow b_0 \to \beta$.
We refer to $\incr$~and~$\incrst$ as \emph{classical comparisons}.
Note that $\alpha \leq \beta$ and $\alpha < \beta$ are always propositional, while $\alpha \incr \beta$ and $\alpha \incrst \beta$ may not be.

We have the following easy but useful observations:

\begin{lemma}[\flink{Lemma-44}]\label{lem:escardo-generalised}
	Simulations induce order-preserving functions,
	\begin{enumerate}[label*=(\roman*)]
		\item $(\alpha \leq \beta) \to (\alpha \incr \beta)$
		\item \label{item:cl-implies-normal} $(\alpha < \beta) \to (\alpha \incrst \beta)$.
	\end{enumerate}
	A bounded order-preserving function is order-preserving,
	\begin{enumerate}[resume,label*=(\roman*)]
		\item $\alpha \incrst \beta \to \alpha \incr \beta$.
	\end{enumerate}
	Classical comparisons have the following transitivity properties:
	\begin{enumerate}[resume,label*=(\roman*)]
		\item $(\alpha \incrst \beta) \to (\beta \incr \gamma) \to (\alpha \incrst \gamma)$
		\item \label{item:cl-important} $(\alpha \incr \beta) \to (\beta \incrst \gamma) \to (\alpha \incrst \gamma)$
	\end{enumerate}
	Finally, $\incrst$ is irreflexive,
	\begin{enumerate}[resume,label*=(\roman*)]
		\item \label{item:cl-irreflexive} $\neg (\alpha \incrst \alpha)$.
	\end{enumerate}
\end{lemma}
\begin{proof}
	The first five observations are trivial.
	For irreflexivity, assume there is an order preserving function $f : \alpha \to \alpha$ and $x_0 : \alpha$ such that $f\, x < x_0$ for all $x : \alpha$.
	This gives rise to an infinitely descending sequence $x_0 > f \, x_0 > f(f \, x_0) > \ldots$, which is impossible since $<$ is wellfounded.
\end{proof}

Regarding \ref{item:cl-important}, note that we do \emph{not} in general have the property $(\alpha \leq \beta) \to (\beta < \gamma) \to (\alpha < \gamma)$, a limitation that motivated Taylor's \emph{plump ordinals}~\cite{taylor1996intuitionistic}.
Further, note that \cref{lem:escardo-op-map} is an immediate consequence of \cref{lem:escardo-generalised}: The assumption $\beta < \alpha$ implies $\beta \incrst \alpha$ by \ref{item:cl-implies-normal}, and the combination with the assumption $\alpha \incr \beta$ implies $\alpha \incrst \alpha$ by \ref{item:cl-important}, a contradiction by \ref{item:cl-irreflexive}.
Vice versa, \cref{lem:escardo-op-map} also implies \cref{lem:escardo-generalised} rather directly; our Agda formalization makes use of this.

\subsection{Functions Specified by Cases}\label{sec:satisfying-cases}

A joint generalization of the arithmetic operations (with one argument fixed) is the following:

\begin{definition}[Specification by Cases \flink{Definition-45}]\label{def:satisfying-case-spec}
	We say that a function $F : \Ord \to \Ord$ is \emph{specified by cases} if there is an ordinal $Z : \Ord$ and an operation $S : \Ord \to \Ord$ such that $F$ satisfies, for all ordinals $\beta$ and all families $L : I \to \Ord$ indexed over a small set $I$:
	\begin{align}
		F(\beta + \One) &= S(F \, \beta) \label{eq:S-condition} \\
		F(\sup_{i : I} L_i) &= Z \vee \sup_{i:I} F(L_i)   \label{eq:sup-spec-case}
	\end{align}
\end{definition}

The definition ensures $F \, \Zero = Z$, so $Z$ is the ``zero case'', while we think of $S$ as the ``successor case''. The ``limit case'' corresponds to the expectation that the function $F$ is \emph{almost} continuous (preserves suprema of inhabited families).
Addition~$\alpha + (-)$, multiplication~$\alpha \times (-)$, and exponentiation~$\alpha^{(-)}$ are all of this form:
\begin{itemize}
	\item For addition, we set $Z \defeq \alpha$ and $S\, \beta \defeq \beta + \One$.
	\item For multiplication, we set $Z \defeq \Zero$ and $S\, \beta \defeq \beta + \alpha$.
	\item For exponentiation, provided that \(\One \le \alpha\), we set $Z \defeq \One$ and $S\, \beta \defeq \beta \times \alpha$.
\end{itemize}

\begin{remark}[\flink{Remark-46}]\label{rem:f-unique}
  In \cref{def:satisfying-case-spec}, note that $Z$ and $S$ determine $F$ uniquely, because, thanks to \cref{ordinal-as-sup-of-successors}, any $F$ which satisfies the definition will be equal to the function $G : \Ord \to \Ord$ defined by transfinite recursion as follows:
  \begin{equation}\label{eq:def-of-F}
    G \, \beta \defeq Z \vee \sup_{b : \beta}\left(S(G(\beta \downarrow b))\right).
  \end{equation}
  That is, if $F$ is specified by cases, then $F = G$. But unless we already know that such a $F$ exists, $G$~itself is not guaranteed to be specified by cases.
\end{remark}

There are three further possible assumptions of interest on $F$ specified by the cases $Z$ and $S$:
\begin{enumerate}[label=(\arabic* \flinkprime{7}{Assumption-\arabic*}), ref=\arabic*]
	\item \label{assumption:one}
	$S$ is a sum of the identity function and another function; i.e., we have a function $H$ such that, for all $\beta$, we have $S\, \beta = \beta + H\, \beta$.
	Simplifying the notation, we write $S = \Id + H$.
	\item \label{assumption:two}
	In addition to the above condition $S = \Id + H$, we also have that, for all $\beta$, $H(F \beta) > \Zero$.
	\item \label{assumption:three}
	Independently of the other assumptions, we can ask that $S$ is ``classically monotone'': If there is an order preserving function from $\beta$ to $\gamma$, then there is an order preserving function from $S\,\beta$ to $S\gamma$, as in
	\begin{equation*}
		(\beta \incr \gamma) \to (S\,\beta \incr S\,\gamma).
	\end{equation*}
\end{enumerate}
Note that \cref{assumption:three} is neither weaker nor stronger than usual monotonicity, but \cref{assumption:three} is exactly monotonicity of $S$ if $\LEM$ is assumed.

Let us analyze whether the assumptions are satisfied for the arithmetic operations.
For addition, this is trivially the case.
For multiplication, the first and the third assumption are always satisfied and the second is equivalent to $\alpha \geq \One$.
For exponentiation, the third assumption is always fulfilled.
If~$\alpha$~has a \emph{trichotomous} least element
($\alpha = {\One + {\alpha_{> \bot}}}$), the first assumption holds too, with the
following proof.
Recall that for exponentiation we have fixed the function $S$ to be
$S\,\beta \defeq \beta \times \alpha$. We need to give a function $H$ such that
$S\,\beta = \beta + {H\,\beta}$. The desired function is
$H\,\beta \defeq \beta \times \alpha_{> \bot}$, because
\[
  \beta \times \alpha = \beta \times (\One + \alpha_{> \bot}) =
  (\beta \times \One) + (\beta \times \alpha_{> \bot}) =
  \beta + (\beta \times \alpha_{> \bot}).
\]

If we additionally assume $\alpha \geq 2$ (and thus $\alpha_{> \bot} \geq \One$), the second condition also holds.
For this, we are required to prove that
$\Zero < H\pa*{\alpha^\beta} = \alpha^\beta \times \alpha_{> \bot}$ for all ordinals
$\beta$.
By \cref{simulation-facts}\ref{simulations-preserve-initial-segments} and the
assumption that $\alpha_{> \bot} \ge \One$, we have $a : \alpha_{> \bot}$ such
that $\alpha_{> \bot} \initseg a = \One \initseg \star = \Zero$.
Using \cref{initial-segment-of-product}, we then calculate that
\[
  \pa*{\abstrexp{\alpha}{\beta} \times \alpha_{> \bot}} \initseg (\bot,a)
  = \abstrexp{\alpha}{\beta} \times \pa*{\alpha_{> \bot} \initseg a}
    + \pa*{\abstrexp{\alpha}{\beta} \initseg \bot}
  = \pa*{\abstrexp{\alpha}{\beta} \times \Zero} + \Zero = \Zero,
\]
so that $\Zero < \abstrexp{\alpha}{\beta} \times \alpha_{> \bot}$, as desired.

\subsection{Properties of Functions Specified by Cases}

In this subsection, we assume that $F$ is specified by cases with $Z$ and $S$ given as above.
The other three assumptions (\cref{assumption:one,assumption:two,assumption:three}) will be stated explicitly.

\begin{lemma}[\flinkprime{7}{Lemma-47}]\label{lem:F-weakly-mono}
	$F$ is weakly monotone, i.e., $\beta \leq \gamma$ implies $F\,\beta \leq F\,\gamma$.
\end{lemma}
\begin{proof}
	It is standard that continuous functions are monotone, and the clause $Z \vee {}$ in \eqref{eq:sup-spec-case} does not alter the argument.
	We consider the family \(L : \Two \to \Ord\) given by
	\(L_0 \defeq \beta\) and \(L_1 \defeq \gamma\).
	If $\beta \leq \gamma$ then
	\(F\,\gamma = F(\sup_{b : \Two}L_b) = Z \vee \sup_{b : \Two}(F\,L_b) \ge Z
	\vee {F\,\beta} \ge F\,\beta\).
\end{proof}

\begin{lemma}[\flinkprime{7}{Lemma-48}]\label{lem:F-strictly-mono}
	Under \cref{assumption:two},
	$F$ is strictly monotone, i.e., $\beta < \gamma$ implies $F \, \beta < F \, \gamma$.
\end{lemma}
\begin{proof}
	Assume $\beta < \gamma$, i.e., we have $\beta = \gamma \downarrow c_0$ for some $c_0$.
	Write $h_0$ for the minimal element of $H(F(\gamma \downarrow c_0))$, given by \cref{assumption:two}.
	Starting with \eqref{eq:def-of-F} from \cref{rem:f-unique}, we then calculate:
	\begin{alignat}{9}
		F \, \gamma &\;& = &\; && Z \vee \sup_{c : \gamma} (S(F (\gamma \downarrow c))) \notag \\
		&& = &&& Z \vee (\sup_{c : \gamma} (F (\gamma \downarrow c) + H(F (\gamma \downarrow c)))) \label{eq:step-with-Z} \\
		&& = &&& \sup_{c : \gamma} (F (\gamma \downarrow c) + H(F (\gamma \downarrow c))) \label{eq:step-wo-Z} \\
		&& > &&& \sup_{c : \gamma} (F (\gamma \downarrow c) + H(F (\gamma \downarrow c))) \downarrow [c_0, \inr\,h_0] \notag \\
                && = &&& F(\gamma \downarrow c_0) + (H(F (\gamma \downarrow c_0)) \downarrow h_0) \notag \\
		&& = &&& F(\gamma \downarrow c_0). \notag
	\end{alignat}
	The step from
	\eqref{eq:step-with-Z} to \eqref{eq:step-wo-Z}
	uses that $Z = F\,\Zero \leq F(\gamma \downarrow c_0) \leq F(\gamma \downarrow c_0) + H(F(\gamma \downarrow c_0)) \leq \eqref{eq:step-wo-Z}$, which allows us to remove $Z$.
\end{proof}

If an ordinal $\beta$ is bounded by ordinals in the image of $F$, then \cref{assumption:one} allows us to significantly tighten the bound to locate $\beta$ in the interval between $F \,\gamma'$ and $F (\gamma' + \One)$ for some $\gamma'$ smaller than the preimage of the original bound:

\begin{lemma}[\flinkprime{7}{Lemma-49}] \label{lem:tightening}
	Let $\beta$ and $\gamma$ be ordinals such that $F \, \Zero \leq \beta < F \, \gamma$.
	Under \cref{assumption:one}, there exists $\gamma' < \gamma$ such that $F \, \gamma' \leq \beta < F (\gamma' + \One)$.
\end{lemma}
\begin{proof}
  We do transfinite induction on $\gamma$.
  Using \eqref{eq:def-of-F} and the assumption $\beta < F\,\gamma$, there exists an~$x$ such that
	\begin{equation*}
		\beta = (Z \vee \sup_{c : \gamma}(S(F(\gamma \downarrow c)))) \downarrow x.
	\end{equation*}
	If $x$ came from the left component ($Z$), giving us some $z : Z$, we would have $\beta = Z \downarrow z < Z \leq F \, \Zero \leq \beta$, which would be a contradiction.
	Therefore, we can assume that $x$ is of the form $\inr\,[c, y]$, with $c : \gamma$ and $y : S(F(\gamma \downarrow c))$, so that we have
	\begin{equation*}
		\beta = S(F(\gamma \downarrow c)) \downarrow y = \left(F(\gamma \downarrow c) + H(F(\gamma \downarrow c))\right) \downarrow y'
	\end{equation*}
	for some $y'$. If $y'$ is of the form $y' = \inr \, h$, we have $F(\gamma \downarrow c) \leq \beta < S(F(\gamma \downarrow c)) = F((\gamma \downarrow c) + \One)$ and are done.
	If $y'$ is of the form $y' = \inl\,y''$, we apply the induction hypothesis to $(\gamma \downarrow c) < \gamma$, since then $\beta = F(\gamma \downarrow c) \downarrow y'' < F(\gamma \downarrow c)$.
\end{proof}

For our main result, we will need to use transfinite induction on pairs of ordinals, with a stronger induction hypothesis than usual. This is justified by the following construction of an ``unordered'' variant of a wellfounded order, where we can either compare with $(x, y)$ or with $(y, x)$. The following result shows that this order is still wellfounded, so that it can be used for transfinite induction.

\newcommand{\uo}{\textup{\textsf{uo}}}
\begin{lemma}[\flinkprime{7}{Lemma-50}]\label{thm:uo-wellfounded}
	Let $(A, <)$ be a wellfounded order.
	The \emph{unordered} order $<_\uo$ on $A \times A$ defined by
	\begin{equation*}
		(a_1, a_2) <_\uo (a'_1, a'_2) \defeq \left(a_1 < a'_1 \wedge a_2 < a'_2\right) + \left(a_1 < a'_2 \wedge a_2 < a'_1\right)
	\end{equation*}
	is also wellfounded.
\end{lemma}
\begin{proof}
  We show that for all $a, b : A$, both $(a,b)$ and $(b,a)$ are $<_\uo$-accessible by $<$-induction on both arguments.
  Given $a$ and $b$, we show that $(a, b)$ is accessible; similarly $(b, a)$ is accessible in a symmetric way.
  We need to show that all predecessors of $(a, b)$ are accessible.
  Such a predecessor is of the form $(a',b')$ or $(b',a')$ with $a' < a$ and $b' < b$, and thus accessible by the induction hypothesis.
\end{proof}

We now aim for the main result of this section, namely that $F\,\beta \leq F\, \gamma$ implies $\beta \leq \gamma$. To get there, we generalise the statement in order for the induction hypothesis to go through: rather than assuming $F\,\beta \leq F\, \gamma$, we assume $F\,\beta \leq F\, \gamma + \delta$ for an ordinal $\delta$ which is not too large, in the sense that $F\, \gamma + \delta < F(\gamma + \One)$. In the end, we can instantiate $\delta = \Zero$ to get the result we want.

\begin{lemma}[\flinkprime{7}{Lemma-51}]\label{thm:cancellation-generalised}
	Under \cref{assumption:one,assumption:two,assumption:three},
	if ${F\, \beta} \leq F \, \gamma + \delta < F(\gamma + \One)$, then $\beta \leq \gamma$.
\end{lemma}
\begin{proof}
  We do wellfounded induction on $(\beta, \gamma)$ using the order $<_\uo$ on $\Ord \times \Ord$ defined in \cref{thm:uo-wellfounded}.
  Fix $b : \beta$.
  By \cref{lem:F-strictly-mono}, we have $F(\beta \downarrow b) < F \, \beta$.
  Combined with the assumed inequalities, this implies $F \, \Zero \leq F(\beta \downarrow b) < F(\gamma+\One)$.
  Thus, there is $x : \gamma + \One$ such that
  \begin{equation}\label{eq:c-cases}
    F((\gamma + \One) \downarrow x) \leq F(\beta \downarrow b) < F((\gamma +\One) \downarrow x + \One)
  \end{equation}
  by \cref{lem:tightening}.

  Let us first assume that $x = \inr\,\star$.
  Then, the first part of \eqref{eq:c-cases} simplifies to $F \, \gamma \leq F(\beta \downarrow b)$, which we claim is impossible, since we then have the following inequality:
  \begin{align*}
    S(F \, \gamma)
    &\incr S(F(\beta \downarrow b)) && \text{(by \cref{assumption:three})} \\
    &= F((\beta \downarrow b) + \One) &&\text{(by the specification)}\\
    &\leq F\, \beta && \text{(by $(\beta \downarrow b) + \One \leq \beta$ and \cref{lem:F-weakly-mono})} \\
    &\leq F\, \gamma + \delta && \text{(by assumption)} \\
    &< F(\gamma + \One) &&  \text{(by assumption)} \\
    &= S(F\, \gamma) && \text{(by the specification)}
  \end{align*}
  Applying \cref{lem:escardo-generalised} \ref{item:cl-important}, this implies $S(F \, \gamma) \incrst S(F \, \gamma)$ and thus a contradiction by \cref{lem:escardo-generalised} \ref{item:cl-irreflexive}.

  Hence, we have $x = \inl \, c_b$, and can write \eqref{eq:c-cases} as:
	\begin{equation}\label{eq:aux}
		F(\gamma \downarrow c_b) \leq F(\beta \downarrow b) < F(\gamma \downarrow c_b + \One).
	\end{equation}
	In particular, we have $F(\gamma \downarrow c_b) \leq F(\beta \downarrow b)$, and thus, using strict monotonicity of $F$,
	\begin{equation}\label{eq:F-inequality}
		F(\gamma \downarrow c_b) \leq F(\beta \downarrow b) < F(\beta \downarrow b + \One).
	\end{equation}
	Hence the induction hypothesis with $\delta = \Zero$ gives us
	\begin{equation}\label{eq:reflection-result-first-part}
		\gamma \downarrow c_b \leq \beta \downarrow b.
	\end{equation}
	(Note that we swapped the arguments of the induction hypothesis here --- this is why we do induction on $<_\uo$.)

	Next, we unfold the last occurrence of $F$ in \eqref{eq:aux} using \cref{assumption:one}:
	\begin{equation*}
		F(\gamma \downarrow c_b) \leq F(\beta \downarrow b) < F(\gamma \downarrow c_b) + H(F(\gamma \downarrow c_b)).
	\end{equation*}
	This implies that there is a $z$ such that
	\begin{equation*}
		F(\beta \downarrow b) = (F(\gamma \downarrow c_b) + H(F(\gamma \downarrow c_b))) \downarrow z.
	\end{equation*}
	If $z$ is of the form $z = \inl\, z_0$, then \(F(\beta \downarrow b) < F(\gamma \downarrow c_b)\) which contradicts \eqref{eq:F-inequality}.
	Hence $z$ must be of the form $z = \inr \, z_1$, which allows us to apply the induction hypothesis, as we can set $\delta \defeq H(F(\gamma \downarrow c_b)) \downarrow z_1$ and get $F(\beta \downarrow b) = F(\gamma \downarrow c_b) + \delta < F(\gamma \downarrow c_b + \One)$.
	Hence we get $\beta \downarrow b \leq \gamma \downarrow c_b$ which combined with \eqref{eq:reflection-result-first-part} gives $\beta \downarrow b = \gamma \downarrow c_b$, and we get $\beta \leq \gamma$ as required.
\end{proof}

By instantiating \cref{thm:cancellation-generalised} with $\delta = \Zero$ and using that $F\,\gamma < F(\gamma + \One)$ by \cref{lem:F-strictly-mono}, we achieve our main result:

\begin{corollary}[\flinkprime{7}{Corollary-52}]
	Under \cref{assumption:one,assumption:two,assumption:three}, if $F\,\beta \leq F\,\gamma$ then $\beta \leq \gamma$ and if $F\,\beta = F\,\gamma$, then~$\beta = \gamma$. \qed
\end{corollary}

As discussed in \cref{sec:satisfying-cases}, addition, multiplication and exponentiation each satisfy \cref{assumption:one,assumption:two,assumption:three}, under some mild and expected conditions.

\begin{theorem}[\flink{Theorem-53}]\label{thm:cancellation}
  For any ordinal $\alpha$, the arithmetic operations reflect $\leq$ and $=$, as follows:
  \begin{enumerate}[label=(\roman*)]
  \item $(\alpha + \_)$, without any condition;
  \item \label{item:cancel-mult} $(\alpha \times \_)$, assuming $\alpha \geq \One$;
  \item $(\alpha^-)$, assuming that $\alpha \geq 2$ and the least element is trichotomous. \qed
  \end{enumerate}
\end{theorem}

It is worth contrasting \cref{thm:cancellation}\ref{item:cancel-mult}
with the situation for mere sets, without ordinal structure. A
classical result by Bernstein~\cite{Bernstein1905} shows that, using
classical logic, for sets $B$ and $C$, if
$\Two \times B \cong \Two \times C$ then $B \cong C$. Swan~\cite{Swan2018}
showed that this is not constructively provable, by exhibiting a topos
where the statement is false. Combined with \cref{thm:cancellation},
this gives a rather indirect proof that it is not constructively
provable that every set can be well ordered.

Finally, we note that uniqueness of simulations means that we can
characterise the computational behaviour of simulations between
arithmetic operations with the left component fixed:

\begin{proposition}[\flink{Proposition-54-i}]
  Let $\alpha$, $\beta$ and $\gamma$ be ordinals.
  \begin{enumerate}[label=(\roman*)]
  \item \label{item:canonical-addition} If $f : \alpha + \beta \leq \alpha + \gamma$, then there exists $h : \beta \leq \gamma$ with $f(\inl\,a) = \inl\,a$ and $f(\inr\,b) = \inr\,(h\,b)$.
  \item If $\alpha \geq \One$ and $f : \alpha \times \beta \leq \alpha \times \gamma$, then there exists $h : \beta \leq \gamma$ such that $f(a, b) = (a, h(b))$.
  \item If $\alpha \geq \Two$, its least element is trichotomous, and $f : \listexp{\alpha}{\beta} \leq \listexp{\alpha}{\gamma}$, then there exists $h : \beta \leq \gamma$ such that $f\,\ell = \map\,((a, b) \mapsto (a , h\,b))\,\ell$.
  \end{enumerate}
\end{proposition}
\begin{proof}
  We prove \ref{item:canonical-addition}; the other cases follow in the same way. Assume $f : \alpha + \beta \leq \alpha + \gamma$. By \cref{thm:cancellation}, there exists $h : \beta \leq \gamma$, and by \cref{+-x-right-monotone}, there is $\overline{h} : \alpha + \beta \leq \alpha + \gamma$ with the desired computational behaviour. But by the uniqueness of simulations, $\overline{h} = f$, and hence also $f$ must compute in the same way.
\end{proof}

\section{Constructive Taboos}\label{sec:taboos}

We were able to give constructive proofs of many desirable properties of ordinal
exponentiation for both abstract and concrete exponentials, e.g.,
monotonicity in the exponent, or algebraic laws such as \(\alpha^{\beta + \gamma} = \alpha^\beta \times \alpha^\gamma\) and \(\alpha^{\beta \times \gamma} = {(\alpha^{\beta})}^{\gamma}\).
In this section, we explore classical properties that are not possible to prove constructively.
A first example is monotonicity in the base, which is
inherently classical.

\begin{proposition}[\flink{Proposition-55}]
  Exponentiation is monotone in the base if and only if \LEM\ holds.
  In fact, \LEM\ is already implied by each of the following weaker statements,
  even when \(\alpha\) and \(\beta\) are each assumed to have a trichotomous least
  element:
  \begin{enumerate}[label=(\roman*)]
  \item\label{exp-taboo} $\alpha < \beta \to \alpha^\gamma \leq \beta^\gamma$;
  \item\label{mult-taboo}
    \(\alpha < \beta \to \alpha \times \alpha \leq \beta \times \beta\).
  \end{enumerate}
\end{proposition}
\begin{proof}
  Note that \ref{exp-taboo} implies \ref{mult-taboo} by taking
  \(\gamma \defeq \Two\).
  To see that \ref{mult-taboo} implies \LEM, we consider an arbitrary
  proposition \(P\) and the ordinals \(\alpha \defeq \Two\) and
  \(\beta \defeq \Three + P\). Clearly, \(\alpha < \beta\) and \(\alpha\)
  and~\(\beta\) respectively have trichotomous least elements \(0\) and
  \(\inl 0\), so that we get a simulation
  \(f : \alpha \times \alpha \to \beta \times \beta\) by assumption.

  If we have \(p : P\), then
  \(g_p : \alpha \times \alpha \to \beta \times \beta\) defined by
  \begin{align*}
    &g_p(0,0) \defeq (\inl 0,\inl 0),\quad
    g_p(1,0) \defeq (\inl 1,\inl 0), \\
    &g_p(0,1) \defeq (\inl 2,\inl 0),\quad
    g_p(1,1) \defeq (\inr p,\inl 0)
  \end{align*}
  can be checked to be a simulation.

  Since simulations are unique, \(g_p\) must agree with \(f\) in case \(P\)
  holds. We now simply check where \((1,1)\) gets mapped by \(f\): if it is of
  the form \((\inr p,y')\) then obviously \(P\) holds; and if it is of the form
  \((\inl y,y')\), then we claim that \(\lnot P\) holds.
  Indeed, assuming \(p : P\) we obtain
  \((\inl y,y') = f(1,1) = g_p(1,1) \equiv (\inr p,\inl 0)\), which is
  impossible as coproducts are disjoint.
\end{proof}

We note that the argument above works for any operation satisfying the
exponentiation specification (by \cref{basic-facts-from-spec}).
The following also works for any operation satisfying \eqref{exp-strong-spec}.

\begin{lemma}[\flink{Lemma-56}]\label{exponentiate-by-prop}
  For a proposition \(P\) we have \(\abstrexp{\Two}{P} = \One + P\).
\end{lemma}
\begin{proof}
  Given \(p : P\), we note that
  \(\abstrexp{\Two}{P \initseg p} \times \Two = \abstrexp{\Two}{\Zero} \times
  \Two = \Two\), so that \(\abstrexp{\Two}{P} = \sup F_{\sbullet}\) with
  \(F : {\One + P} \to \Ord\) defined as \(F(\inl \star) \defeq \One\) and
  \(F(\inr p) \defeq \Two\).
  Since \(\sup\) gives the least upper bound and \(P\) implies
  \(\One + P = \Two\), we get a simulation \(\sup F_{\sbullet} \le \One + P\).
  Conversely, we note that
  \begin{align*}
    &(\One + P) \initseg (\inl\hspace{1pt} \star) = \Zero = \One \initseg \star =
    \sup F_{\sbullet} \initseg [\inl\hspace{1pt} \star,\star] \text{, and} \\
    &(\One + P) \initseg (\inr p) = \One = \Two \initseg 1 =
    \sup F_{\sbullet} \initseg [\inr p,1],
  \end{align*}
  yielding a simulation \({\One + P} \le {\sup F_{\sbullet}}\).
\end{proof}

As mentioned just after \cref{simulation-facts}, showing that \(\alpha \leq \beta\) is often straightforward in a classical metatheory, as \(\LEM\) ensures that a simulation can always be ``carved out'' out of any order preserving function \(\alpha \to \beta\). This result is unavoidably classical, in the sense that it in turn implies the law of excluded middle.

\begin{lemma}[\flink{Lemma-57}]\label{order-preserving-gives-sim-iff-LEM}
  Every order preserving function between ordinals induces a simulation if and only if\/ \(\LEM\) holds.
\end{lemma}
\begin{proof}
  The right-to-left
  direction --- that \(\LEM\) implies that order preserving functions induce simulations --- was proven and formalized by {Escard\'o}~\cite[\texttt{Ordinals.OrdinalOfOrdinals}]{TypeTopologyOrdinals}: using \(\LEM\), to prove \(\alpha \leq \beta\) it suffices to show that it is impossible that \(\beta < \alpha\), but this readily follows from having an order preserving function from \(\alpha\) to \(\beta\).

  In the other direction, let \(P\) be an arbitrary proposition, and consider \(\alpha \defeq \One\) and \(\beta \defeq P + \One\). The function \(\star \mapsto \inr\,\star : \alpha \to \beta\) is trivially order preserving, and thus gives rise to a simulation \(f : \alpha \to \beta\) by assumption. Since simulations preserve least elements, we can then decide \(P\), since \(P\) holds if and only if \(f\,\star = \inl\,p\) for some \(p : P\).
\end{proof}

The following proposition shows that the arithmetical operations are
inflationary if and only if \(\LEM\) holds.

\begin{proposition}[\flink{Proposition-58-i}]\label{inflationary-taboo}
  The following are all equivalent to \(\LEM\) (and hence to each other).
  \begin{enumerate}[label=(\roman*)]
  \item\label{addition-inflationary} For all ordinals \(\alpha\) and \(\beta\), we
    have \(\beta \le \alpha + \beta\).
  \item\label{multiplication-inflationary} For all ordinals \(\beta\) and
    \(\alpha > \Zero\), we have \(\beta \le \alpha \times \beta\).
  \item\label{exp-inflationary-base-two} For all ordinals \(\beta\), we have
    \(\beta \le \abstrexp{\Two}{\beta}\).
  \item\label{exp-inflationary} For all ordinals \(\beta\) and
    \(\alpha > \One\), we have \(\beta \le \abstrexp{\alpha}{\beta}\).
  \end{enumerate}
  Further, $\LEM$ can be derived from each of
  \ref{addition-inflationary}--\ref{exp-inflationary} even if we require the
  ordinals \(\alpha\) and \(\beta\) to each have a trichotomous least element.
\end{proposition}
\begin{proof}
  Assuming \(\LEM\), it suffices to construct order preserving maps to prove all
  inequalities. For \ref{addition-inflationary} the right coproduct inclusion
  is such a map. For \ref{multiplication-inflationary} we can use
  \(b \mapsto (a_0,b)\) where \(a_0\) is any element of \(\alpha\), which exists
  since \(\alpha > \Zero\).
  For \ref{exp-inflationary}, we have \(a_1 : \alpha\) such that
  \({\alpha \initseg a_1} = \One\) and we claim that the map
  \(f\,b \defeq [\inr b , (\bot , a_1)]\) does the job.
  Indeed, by \cref{initial-segment-of-abstrexp}, we have
  \(
    \abstrexp{\alpha}{\beta} \initseg f\,b = \abstrexp{\alpha}{\beta \initseg b},
  \)
  so that if \(b < b'\), then
  \(\abstrexp{\alpha}{\beta \initseg b} < \abstrexp{\alpha}{\beta \initseg b'}\)
  by \cref{abstrexp-monotone} and hence \(f\,b < f\,b'\) as the operation of
  taking initial segments is order reflecting.
  Finally, \ref{exp-inflationary-base-two} follows from
  \ref{exp-inflationary} of course.

  We now show that each of \ref{addition-inflationary},
  \ref{multiplication-inflationary} and \ref{exp-inflationary-base-two}
  implies \(\LEM\). It follows that \ref{exp-inflationary} also implies
  \(\LEM\).

  Assuming \ref{addition-inflationary}, we consider, for an arbitrary proposition
  \(P\), the ordinals \(\alpha \defeq \One + P\) and \(\beta \defeq \Two\).
  By assumption we get a simulation \(f : \Two \le {\One + P + \Two}\).
  Now we check where \(1 : \Two\) is sent by \(f\). Since simulations preserve
  least elements \(0 : \Two\) is already sent into the \(\One\)-component, and
  since simulations are also injective, \(1\) either goes into the
  \(P\)-component or into the \(\Two\)-component.
  In the first case, \(P\) must hold. In the second case, \(P\) cannot hold, because
  simulations preserve initial segments so that the equation
  \(
    \Two \initseg 1 = \One = (\One + P + \Two) \initseg f\,1
  \)
  must hold; but the right hand side has too many elements if \(P\) holds and
  \(f\,1\) is in the \(\Two\)-component.

  Assuming \ref{multiplication-inflationary}, we consider, for an arbitrary proposition
  \(P\), the ordinals \(\alpha \defeq {\One + P}\) and \(\beta \defeq \Two\).
  By assumption we get a simulation \(f : \Two \le {(\One + P) \times
    \Two}\). Using that simulations preserve initial segments, we can decide \(P\)
  by inspecting where \(f\,1\) gets mapped.
  If it gets mapped to an element of the form \((\inr p,b)\) then obviously
  \(P\) holds.
  If it gets mapped to an element of the form \((\inl \star,b)\), then \(b = 1\)
  as \(f\,0 = (\inl \star,0)\) and simulations are injective (by
  \cref{simulation-facts}).
  Now if \(P\) were true, then \(f\,1\) would have two predecessors, namely
  \((\inl \star,0)\) and \((\inr p,0)\), which is impossible as \(1\) has only
  one. Thus, \(\neg P\) holds in this case.

  Given \ref{exp-inflationary-base-two} and a proposition \(P\), we set \(\beta \defeq \One + P + \One\).
  Using \cref{exponentiate-by-prop} and \cref{abstrexp-by-+} we have
  \(\abstrexp{\Two}{\beta} = \Two \times (1 + P) \times \Two = (\Two + \Two
  \times P) \times \Two\), so that we get a simulation
  \(f : {\One + P + \One} \le {(\Two + \Two \times P) \times \Two}\) by assumption.
  We now decide \(P\) by inspecting \(y \defeq f(\inr \star)\).

  Since simulations preserve least elements, we note that $y = (\inl 0,0)$ is
  impossible.
  If \(y\) is of the form \((\inr(\_,p),\_)\), then obviously \(P\) holds.
  If \(y = (\inl 1,0)\), then \(y\) is the second element of the codomain of the
  simulation, so that, by an argument similar to the one given for
  \ref{addition-inflationary}, the element \(\inr \star\) must be the
  second element of the domain, and hence \(\neg P\) must hold.
  If \(y\) is of the form \((\inl b,1)\), then by the simulation property and
  the fact that \((\inl 1,0) < (\inl b,1)\), we have \(x : {\One + P + \One}\)
  with \(f\,x = (\inl 1,0)\). But since \(x = (\inl(\inl\star))\) and
  \(x = \inr \star\) are impossible, \(x\) must yield an element of~\(P\).

  Finally, we note that all our choices of ordinals for \(\alpha\) and \(\beta\)
  have trichotomous least elements.
\end{proof}

\section{Approximating Subtraction, Division and Logarithm Operations}\label{sec:approximation}

It is natural to wonder whether we can develop constructive ordinal arithmetic
further by constructing subtraction, division and logarithm operations.
Like with exponentiation, a careful treatment is required here, as e.g., the
existence of ordinal subtraction in the usual formulation is equivalent to
excluded middle, as observed by
Escard\'o~\cite[{\texttt{Ordinals.AdditionProperties}}]{TypeTopologyOrdinals}.
As we will see, it is possible, however, to construct a function
\(f : \Ord \times \Ord \to \Ord\) satisfying the weaker requirement that
\(f(\alpha,\beta)\) is the greatest ordinal \(\gamma\) with
\(\alpha + \gamma \leq \beta\) and \(\gamma \leq \beta\), provided \(\alpha \leq \beta\).

We start by proving that every bounded and antitone predicate that is closed
under suprema has a greatest element satisfying it.

\begin{definition}[Antitone, bounded predicates closed under suprema \flink{Definition-59}]
  A type family \(P\) over ordinals is
  \begin{enumerate}[label=(\roman*)]
  \item \emph{antitone} if \(\alpha \le \beta\) and \(P\,\beta\) together imply
    \(P\,\alpha\);
  \item \emph{bounded} if there exists an ordinal \(\delta\) such that for all
    ordinals \(\alpha\) with \(P\,\alpha\), we have \(\alpha \le \delta\);
  \item \emph{closed under suprema} if for every family \(F\) of ordinals,
    indexed by a small type \(I\), we have \(P\,(\sup F_{\sbullet})\) whenever
    we have \(P\,(F\,i)\) for all \(i : I\).
  \end{enumerate}
\end{definition}

We are interested in predicates which are antitone, bounded and closed under suprema, since these conditions guarantee that there is a greatest element satisfying the predicate.
For a proposition-valued predicate $P$, greatest ordinals are unique,
if they exist, hence we can formulate ``has a greatest ordinal satisfying $P$''
equivalently using \(\Sigma\) or \(\exists\); the resulting type will always be a proposition.

\begin{proposition}[\flink{Proposition-60}]\label{greatest-ordinal-satisfying-predicate}
  Any antitone, bounded, proposition-valued (and small-valued) predicate \(P\) closed under suprema has a greatest ordinal satisfying it.
\end{proposition}
\begin{proof}
  Since we are proving a proposition, we can assume that the bound
  \(\delta\) is explicitly given to us.
  Let us write \({S\,d} \defeq {\delta \initseg d} + \One\) and consider
  \[
    \gamma \defeq \sup_{d : \delta , P\,(S\,d)}S\,d
  \]
  Since \(P\) is closed under suprema, we have \(P\,\gamma\) by definition of
  \(\gamma\).
  Now suppose that we have \(\alpha\) with~\(P\,\alpha\). We want to show that
  \(\alpha \le \gamma\). Since \(\delta\) is a bound for \(P\), we have, for every
  \(a : \alpha\), an element \(d_a : \delta\) with
  \({\alpha \initseg a} = {\delta \initseg d_a}\).
  Since we have \({\alpha \initseg a + \One} \le \alpha\) and \(P\,\alpha\), we
  get \(p_a : P\,(S\,d_a)\) because \(P\) is antitone.
  Hence, \(c_a \defeq [(d_a,p_a),\inr \star]\) is an element of \(\gamma\) for
  which we compute, using~\cref{initial-segment-of-sup}, that
  \(\gamma \initseg c_a = {S\,d_a \initseg \inr \star} = {\delta \initseg d_a} =
  {\alpha \initseg a}\).
  Thus, the assignment \(a \mapsto c_a\) defines a simulation from~\(\alpha\) to~\(\gamma\), as desired.
\end{proof}

We remark that the result even holds for ``predicates'' which are not necessarily proposition-valued, as long as the bound is concretely given (i.e., if boundedness is formulated using a \(\sum\)-type rather than with its propositional truncation \(\exists\)).

\begin{theorem}[\flink{Theorem-61-i}]\label{Enderton-like}
  Let \(t : \Ord \to \Ord\) be an endofunction on ordinals that preserves suprema up to a binary
  join, in the sense that we have an ordinal \(\delta_0\) such that for every
  small family \(F\) of ordinals, the equation
  \[
    t (\sup F_{\sbullet}) = \delta_0 \vee {\sup (t \circ F_{\sbullet})}
  \]
  holds. Then for any ordinal \(\delta\) with \(\delta_0 \le \delta\), we have a
  greatest ordinal \(\gamma\) such that \(\gamma \le \delta\) and
  \(t\,\gamma \le \delta\).

  Moreover, if \(t\) is inflationary, i.e.\ \(\alpha \le t\,\alpha\) holds for
  all \(\alpha\), then we have an ordinal \(\gamma \le \delta\)
  such that \(\gamma\) is the greatest ordinal with \(t\,\gamma \le \delta\).
\end{theorem}

Before proving the theorem, we pause to discuss it.
Note that the equation forces \(t\,\Zero = \delta_0\) by considering the
supremum of the empty family.
For a fixed ordinal \(\alpha\), examples of such endofunctions are
\begin{itemize}
\item addition \(\alpha + (-)\) with \(\delta_0 = \alpha\),
\item multiplication \(\alpha \times (-)\) with \(\delta_0 = \Zero\),
\item exponentiation \(\abstrexp{\alpha}{(-)}\) with \(\delta_0 = \One\) (for \(\alpha \ge \One\)).
\end{itemize}

The theorem is close to Enderton's~\cite{Enderton1977} classical Theorem Schema 8D, but with two differences:
\begin{enumerate}[label=(\roman*)]
\item \emph{loc.~cit.}\ restricts to ``normal'' operations, i.e., endomaps \(t : \Ord \to \Ord\)
  such that \(t\) preserves the strict order \(<\), and
  \(t\,\lambda = \sup_{\beta < \lambda}t\,\beta\) for limit ordinals
  \(\lambda\);
\item \emph{loc.~cit.}\ proves: for any bound \(\delta\) with
  \(\delta_0 \le \delta\), we have a greatest ordinal \(\gamma\) such that
  \(t\,\gamma \le \delta\) (so the condition \(\gamma \le \delta\) is absent).
\end{enumerate}
We will shortly see that in several examples of \(t\) and \(\delta\), excluded
middle is equivalent to the existence of \(\gamma\) such that
\(\gamma \le \delta\) and \(\gamma\) is the greatest ordinal such that
\(t\,\gamma \le \delta\).

\begin{proof}[Proof of~\cref{Enderton-like}]
  Given such an endofunction \(t\) and ordinals \(\delta_0\) and \(\delta\), we
  use \cref{greatest-ordinal-satisfying-predicate} with the family
  \(P\,\alpha \defeq (t\,\alpha \le \delta) \wedge (\alpha \le \delta)\).
  Then \(P\) is closed under suprema because \(t\)~preserves suprema up to a
  join with \(\delta_0\), antitone because \(t\) is monotone (this follows from
  the preservation of suprema up to a join with \(\delta_0\)) and bounded by
  \(\delta\).
  Hence we obtain a greatest ordinal~\(\gamma\) satisfying \(P\) which proves
  the first part of the theorem.

  Now suppose in addition that \(t\) is inflationary. We claim that in this case
  \(\gamma\) is the greatest ordinal with \(t\,\gamma \le \delta\).
  Indeed, if we have \(\alpha\) with \(t\,\alpha \le \delta\), then
  \(\alpha \le t\,\alpha \le \delta\) as \(t\) is inflationary, so
  \(\alpha \le \gamma\) by construction of \(\gamma\).
\end{proof}

We now consider some examples and applications.
As already mentioned, the existence of ordinal subtraction in the usual
formulation is equivalent to the law of excluded middle, as observed by
Escard\'o~\cite[{\texttt{Ordinals.AdditionProperties}}]{TypeTopologyOrdinals}.
It is possible, however, to approximate ordinal subtraction, as well as division
and logarithms of ordinals, in the following sense.

\begin{proposition}[\flink{Proposition-62}]\label{thm:approximation}
  For any two ordinals \(\alpha\) and \(\beta\), we have greatest ordinals
  \(\gamma_s\), \(\gamma_d\) and \(\gamma_l\) such that
  \begin{enumerate}[label=(\roman*)]
  \item\label{subtraction} \(\gamma_s \le \beta\) and \(\alpha + \gamma_s \le \beta\),
    provided \(\alpha \le \beta\);
  \item\label{division} \(\gamma_d \le \beta\) and
    \(\alpha \times \gamma_d \le \beta\);
  \item\label{logarithm} \(\gamma_l \le \beta\) and
    \(\abstrexp{\alpha}{\gamma_l} \le \beta\), provided \(\beta \ge \One\).
  \end{enumerate}
\end{proposition}
\begin{proof}
  To prove \ref{subtraction}, apply \cref{Enderton-like} with the parameters
  \(t = \alpha + (-)\), \(\delta_0 = \alpha\) and \(\delta = \beta\), and use
  \cref{+-preserves-suprema}.
  For \ref{division}, apply \cref{Enderton-like} with the parameters
  \(t = \alpha \times (-)\), \(\delta_0 = \Zero\) and \(\delta = \beta\), and use
  \cref{multiplication-is-continuous}.
  Finally, for \ref{logarithm}, apply \cref{Enderton-like} with the parameters
  \(t = \abstrexp{\alpha}{(-)}\), \(\delta_0 = \One\) and \(\delta = \beta\).
\end{proof}

Note that it is not technically necessary to assume \(\alpha > \Zero\)
in~\ref{division}, even though division by \(\Zero\) is not well defined. In
fact, the \(\gamma_d \le \beta\) requirement forces \(\gamma_d = \beta\) in case
\(\alpha = \Zero\).
Similarly, we do not need to assume \(\alpha > \One\) in~\ref{logarithm}.

Perhaps surprisingly, seemingly mild variations of the above are equivalent to
\(\LEM\).
\begin{proposition}[\flink{Proposition-63-i}]
  The following are all equivalent to \(\LEM\).
  \begin{enumerate}[label=(\roman*)]
  \item\label{subtraction-var} For every \(\alpha\), \(\beta\) with \(\alpha \le \beta\), we have
    \(\gamma \le \beta\) such that \(\gamma\) is the greatest with
    \(\alpha + \gamma \le \beta\).
  \item\label{division-var} For every \(\alpha\) and \(\beta\) with
    \(\alpha > \Zero\), we have \(\gamma \le \beta\) such that \(\gamma\) is the
    greatest with \(\alpha \times \gamma \le \beta\).
  \item\label{logarithm-var} For every \(\alpha\) and \(\beta\) with
    \(\alpha > \One\) and \(\beta \ge \One\), we have \(\gamma \le \beta\) such that
    \(\gamma\) is the greatest with \(\abstrexp{\alpha}{\gamma} \le \beta\).
  \end{enumerate}
\end{proposition}
\begin{proof}
  We only spell out the details for \ref{subtraction-var}. The equivalence of
  \ref{division-var} and \ref{logarithm-var} with \(\LEM\) is proved
  similarly, using the other clauses of \cref{inflationary-taboo}.
  If \(\LEM\) holds, then the operation \(\alpha + (-)\) is inflationary by
  \cref{inflationary-taboo}. Hence we can apply the
  second clause of \cref{Enderton-like} to prove \ref{subtraction-var}.
  Conversely, assume we have \ref{subtraction-var}, we show that
  \(\beta \le \alpha + \beta\) holds for all ordinals \(\alpha\) and \(\beta\),
  which implies \(\LEM\) by \cref{inflationary-taboo}.
  Given \(\alpha\) and \(\beta\), we use \ref{subtraction-var} for \(\alpha\)
  and \(\alpha + \beta\) to obtain \(\gamma \le {\alpha + \beta}\) such that
  \(\gamma\) is the greatest with \(\alpha + \gamma \le {\alpha + \beta}\).
  Obviously \({\alpha + \beta} \le {\alpha + \beta}\), so that
  \(\beta \le \gamma\), and hence \(\beta \le {\alpha + \beta}\), completing the
  proof.
\end{proof}

\section{Conclusions and Future Work}

We have presented two constructively well behaved
definitions of ordinal exponentiation and showed them to be equivalent in case
the base ordinal has a trichotomous least element.
Working in homotopy type theory, we used the equivalence, in combination with the univalence axiom, to transfer
various results, such as algebraic laws and decidability properties, from one
construction to the other.
We also showed that arithmetic operations can be constructively cancelled on the left, e.g.\ \(\alpha \times \beta \leq \alpha \times \gamma\) implies \(\beta \leq \gamma\) for \(\alpha \geq \One\), and
we furthermore marked the limits of a constructive treatment by presenting no-go
theorems that show the law of excluded middle (\LEM) to be equivalent to certain
statements about ordinal exponentiation.

A natural question, to which we do not yet have a conclusive answer, is whether
it is possible to fuse the two constructions of this paper and define ordinal
exponentiation for base ordinals that do not necessarily have a trichotomous
least element via quotiented lists.
Other worthwhile open questions include what kind of large ordinals can be shown to exist in constructive foundations, now that we have access to constructive ordinal exponentiation, and if these larger ordinals can be useful for applications such as the construction of fixed points by transfinite iteration.

\begin{acks}
  We are grateful to Mart\'in Escard\'o, Paul Levy and David W\"arn for helpful
  discussions on ordinals. In particular, Escard\'o pointed us to the work of
  Grayson, while the definition of abstract exponentiation builds on a
  suggestion by W\"arn.
  We appreciate the pointers and suggestions by the anonymous reviewers that
  helped us to improve the paper.
  In particular, we are thankful for a question by one of the reviewers, which led us to strengthen \cref{inflationary-taboo}.

  We also thank the participants, organizers and support staff of the HIM
  trimester programme \emph{Prospects of Formal Mathematics} for the opportunity
  to work on this project during our stay.

  This work was supported by The Royal Society (grant references
  URF\textbackslash{}R1\textbackslash{}191055,
  RF\textbackslash{}ERE\textbackslash{}210032,
  RF\textbackslash{}ERE\textbackslash{}231052,
  URF\textbackslash{}R\textbackslash{}241007), the UK National Physical
  Laboratory Measurement Fellowship project ``Dependent types for trustworthy
  tools'', and the Engineering and Physical Sciences Research Council
  (EP/Y000455/2).
  The HIM programme was funded by the Deutsche Forschungsgemeinschaft (DFG,
  German Research Foundation) under Germany's Excellence Strategy -- EXC-2047/1
  -- 390685813.
\end{acks}

\printbibliography

\end{document}